\ifx\onecol\undefined
\documentclass[journal,twocolumn]{IEEEtran}
\else 
\documentclass[12pt,draftclsnofoot,journal,onecolumn]{IEEEtran}
\fi

\usepackage{amsmath,amsfonts}
\usepackage{algorithmic}
\usepackage{algorithm}
\usepackage{array}
\usepackage{color}
\usepackage{textcomp}
\usepackage{stfloats}
\usepackage{url}
\usepackage{verbatim}
\usepackage{graphicx}
\graphicspath{{Figures/}}
\usepackage{cite}
\usepackage{amsthm}
\usepackage{gensymb}
\usepackage{bm}
\usepackage{setspace}
\usepackage{hyperref}
\usepackage{subfigure}
\usepackage{makecell}
\hypersetup{hidelinks, 
colorlinks=true,
allcolors=black,
pdfstartview=Fit,
breaklinks=true}
\usepackage{arydshln}

\newtheorem{proposition}{\bf Proposition}

\begin{document}

\title{Coded Beam Training for RIS Assisted Wireless Communications}

\author{Yuhao Chen, \textit{Graduate Student Member, IEEE} and Linglong Dai, \textit{Fellow, IEEE}
\thanks{This paper was supported by National Key Research and Development Program of China (Grant No. 2023YFB3811503). \textit{(Corresponding author: Linglong Dai.)}}
\thanks{The authors are with the Department of Electronic Engineering, Tsinghua University, Beijing 100084, China, and also with the Beijing National Research Center for Information Science and Technology (BNRist), Beijing 100084, China. (e-mails: \href{mailto:chen-yh21@mails.tsinghua.edu.cn}{chen-yh21@mails.tsinghua.edu.cn}; \href{mailto:daill@tsinghua.edu.cn}{daill@tsinghua.edu.cn}).}
}

\maketitle

\begin{abstract}

Reconfigurable intelligent surface (RIS) is considered as one of the key technologies for future 6G communications. To fully unleash the performance of RIS, accurate channel state information (CSI) is crucial. Beam training is widely utilized to acquire the CSI. However, before aligning the beam correctly to establish stable connections, the signal-to-noise ratio (SNR) at UE is inevitably low, which reduces the beam training accuracy. To deal with this problem, we exploit the coded beam training framework for RIS systems, which leverages the error correction capability of channel coding to improve the beam training accuracy under low SNR. Specifically, we first extend the coded beam training framework to RIS systems by decoupling the base station-RIS channel and the RIS-user channel. For this framework, codewords that accurately steer to multiple angles is essential for fully unleashing the error correction capability. In order to realize effective codeword design in RIS systems, we then propose a new codeword design criterion, based on which we propose a relaxed Gerchberg-Saxton (GS) based codeword design scheme by considering the constant modulus constraints of RIS elements. In addition, considering the two dimensional structure of RIS, we further propose a dimension reduced encoder design scheme, which can not only guarentee a better beam shape, but also enable a stronger error correction capability. Simulation results reveal that the proposed scheme can realize effective and accurate beam training in low SNR scenarios.

\end{abstract}

\begin{IEEEkeywords}
	RIS, beam training, channel coding, codeword design.
\end{IEEEkeywords}

\vspace{-10pt}

\section{Introduction}
Reconfigurable intelligent surface (RIS) is considered as a promising technology for future 6G wireless communications~\cite{basar2019wireless}. Thanks to the numerous low-cost reflecting elements, RIS can control the electromagnetic environment intelligently with low power consumption~\cite{di2020smart,wu2019towards,pan2021reconfigurable}. By properly controlling the phase shifts of RIS elements, directional beams with high array gain could be generated by beamforming to extend the signal coverage and improve the channel capacity~\cite{zhou2020framework,zhi2022power}. In order to realize effective beamforming so as to fully leverage the potential benefits of RIS, accurate channel state information (CSI) is essential~\cite{zheng2022survey,alwazani2020intelligent}.

The CSI can be obtained by either explicit channel estimation or implicit beam training. For explicit channel estimation, since the elements on RIS can only reflect the incident signals, the base station (BS) needs to estimate the cascaded channel (the composite of user-RIS channel and RIS-BS channel)~\cite{wang2020channel,wei2021channel}. 
For example, authors in~\cite{an2021low} focused on the low-complexity uplink channel estimation in RIS systems and considered the discrete phase shifts on RIS. The impact of different configurations of RIS on the channel estimation accuracy was also analyzed. Furthermore, authors in~\cite{an2022codebook} proposed a codebook-based solution for RIS systems. By selecting proper codebook during uplink channel estimation, the proposed framework could be adaptive to different pilot overhead constraints and has reduced error propagation, control signaling, and computational complexity. The above two works both focused on direct estimation of the superposed channel in RIS systems and neglect the information in the environment. To solve this problem, authors in~\cite{yu2024environment} proposed an environment-aware codebook generation scheme. By generating a group of channels based on the statistical CSI, the configuration at RIS during uplink channel estimation can be optimized to facilitate the channel estimation. Other environment information such as the geometric relationship between RIS and UE was also exploited in~\cite{chen2024channel} to improve the channel estimation accuracy. However, for explicit channel estimation, the size of the cascaded channel is the product of the number of BS antennas and the number of RIS elements. With the large number of RIS elements needed to generated high-gain beams, the size of cascaded channel is usually large, leading to an unacceptable pilot overhead. To avoid estimating the large cascaded channel matrix, the implicit beam training, which only aims to determine the angles of RIS and user equipment (UE), can be utilized. By searching the space with a series of pre-defined codewords, the angles of RIS and UE can be estimated based on the received power, according to which the beams at BS and RIS can be correctly aligned to UE.
\vspace{-10pt}
\subsection{Prior Works}
\vspace{-3pt}

To determine the angles of RIS and UE, the most intuitive scheme is the exhaustive beam training~\cite{suh2016construction,chen2023accurate}. During beam training, BS and RIS both generate narrow beams to sequentially search all possible angles in space. After transmitting all candidate beams, the angles of RIS and UE can be obtained by selecting the beams with the maximum received power. For this scheme, both BS and RIS generate high-gain narrow beams, so the angles can be accurately estimated. However, for this scheme, the number of candidate beams equals to the product of the number of BS antennas and the number of RIS elements. With the large number of BS antennas and RIS elements, there is a very large number of candidate beams, resulting in an overwhelming beam training overhead.

In order to reduce the beam training overhead, researchers have developed various low overhead hierarchical beam training schemes. Existing works can be divided into two categories. The first category is single user hierarchical beam training, where the beam training overhead can be reduced by excluding a large range of impossible angles so as to narrow down the searching range effectively. 

For example, in~\cite{wei2022codebook}, beams generated by lower-layer codewords cover a wider range of angles compared with beams generated by higher-layer codewords. During beam training, lower-layer codewords are firstly transmitted. After transmitting the codewords in a certain layer, the index of the beam with the maximum received power is fed back to the BS and RIS, and the codewords in the next layer are decided accordingly. As the layer grows higher, the searching range gradually narrows down and the angular resolution increases continuously. After the highest-layer search, the angles of RIS and UE are then determined. By this means, a large range of wrong angles is excluded in lower-layer search, which avoids a lot of unnecessary high-resolution search and thus reduces the beam training overhead. However, since the higher-layer codewords is determined by the result of lower-layer training, this category of schemes requires frequent feedback between UE and BS/RIS, which brings extra burden to the RIS systems. Moreover, since the searching range in a certain layer may vary for different UEs, it is hard to extend this category of schemes to multi-user systems, which severely limits the application of such schemes. 

To realize effective beam training in multi-user systems, for the second category, all possible angles in space are divided into several disjoint subsets, and the beam generated by each codeword covers the angles in a certain subset simultaneously. After scanning the entire space by these codewords, the subset of the beam with the maximum received power is recorded. Then, the space is divided in a different way and the corresponding scanning and recording are conducted again. After a few rounds of scanning, the angles of RIS and multiple UEs can be determined independently based on the intersections of all recorded subsets. 

Following this idea, researchers have studied several effective beam training schemes~\cite{wang2021fast,you2020fast,xu2024low,wang2023hierarchical}. Specifically, authors in~\cite{wang2021fast,you2020fast} divided the whole space for every round of scanning in a random/hashed way. Authors in~\cite{xu2024low} further extend this hashing scheme to multi-RIS scenarios. By assigning different hashing functions to different RISs, the angles of different RISs can be simultaneously determined. For the above schemes, the choice of hashing function may affect the beam training accuracy, thus leading to an unstable performance. To deal with this, authors in~\cite{wang2023hierarchical} studied a full-coverage hierarchical beam training scheme. Unlike the single user hierarchical beam training of the first category, in each layer, the beams also cover the whole space, but the angular resolution gradually increases. For this category of schemes, each beam can search multiple angles simultaneously, so the beam training overhead can also be reduced. In addition, since how to divide the entire space in each round does not depend on previous results, no extra feedback is needed during beam training, which makes this category of schemes more adaptive to different communication scenarios than the first category. 

Unfortunately, before aligning the beam correctly to establish stable connections among BS, RIS and UE, the signal-to-noise ratio (SNR) at UE is inevitably low. What's worse, in RIS systems, there exists the “multiplicative fading” effect~\cite{zhang2022active,zhi2022active}, which means the equivalent path loss of the BS-RIS-UE link is the product of the path loss of BS-RIS link and the path loss of RIS-UE link. Meanwhile, both categories of low overhead beam training schemes need to generate beams that cover a large range of angles, leading to a relatively low beamforming gain. These facts will result in a significantly low SNR at UE. As a result, the codeword may be mischosen, which leads to the “error propagation” phenomenon and greatly reduces the beam training accuracy. Therefore, how to realize accurate beam training in RIS systems under poor SNR conditions is crucial for the practical deployment of RIS in future communications.

\subsection{Our Contributions}

To improve the beam training accuracy under poor SNR scenarios, in this paper, we exploit the coded beam training framework in RIS systems. By applying the idea of channel coding in the beam training process, we can leverage the error correction capability of channel coding to enhance the reliability of beam training under low SNR\footnote{Simulation codes will be provided to reproduce the results in this article: \href{http://oa.ee.tsinghua.edu.cn/dailinglong/publications/publications.html}{http://oa.ee.tsinghua.edu.cn/dailinglong/publications/publications.html}.}. The specific contributions are listed as follows.

\begin{itemize}
	\item First, inspired by the coded beam training framework that is recently studied by us in multiple-input multiple-output (MIMO) systems~\cite{zheng2024coded}, we design a coded beam training framework for RIS systems. Specifically, we map the angles in space to different beam patterns in space through the encoding function. Based on the intended beam patterns, we design the codewords, which is the foundation of the designed framework. After sequentially transmitting all codewords in the codebook, the UE can obtain the encoded transmitting sequence based on the received powers. Then, the decoding function is utilized to decode the received sequence and the angles of RIS and UE can be estimated. Thanks to the error correction capability of the encoding-decoding process, the error caused by low SNR during beam training can be corrected and the beam training accuracy improves accordingly.

	\item The most significant difference between the codeword design of RIS and the codeword design for the coded beam training framework in~\cite{zheng2024coded} is that, RIS is subject to constant modulus constraint, which makes it hard to generate ideal beams that cover a variety of angles. One of the efficient codeword design schemes is Gerchberg-Saxton (GS) based codeword design scheme~\cite{lu2023hierarchical}. To adapt the GS-based scheme to RIS, we first clarify that the criterion of minimizing the difference between the intended beam shape and the generated beam shape is actually unsuitable for beam training. Then, we propose that the criterion for codeword design should be distinguishing between the angles within the intended angle coverage range and the angles out of the angle coverage range. Based on this new criterion, we propose a relaxed GS-based codeword design scheme so as to improve the beam shape accuracy.
	
	\item Apart from the constant modulus constraint, the structure of RIS is usually a 2-dimensional (2D) uniform planar array (UPA), which is also different from the uniform linear array (ULA) considered in~\cite{zheng2024coded}. The 2D structure leads to a poor orthogonality for the spatial sampling matrix in the proposed relaxed GS-based codeword design scheme, which also results in a non-ideal beam shape. To deal with this problem, we further propose a dimension reduced encoder design scheme. By decoupling the 2D codeword design problem into two 1D codeword design problems, the spatial sampling matrix degenerates to the 1D case and possesses a good orthogonality, thus improving the quality of the beam shape design. Moreover, since the encoder decouples the two dimensions of RIS, the error correction capability can also be improved. Then, we compare the necessary beam training overheads of the proposed framework and existing frameworks. Finally, simulation results reveal that the proposed framework can realize efficient beam training in low SNR scenarios.

\end{itemize}

\subsection{Organization and Notation}
\subsubsection*{Organization}
The remainder of this paper is organized as follows. In Section II, we first introduce the system model. Then, the traditional exhaustive beam training framework and hierarchical beam training framework are elaborated. In Section III, we introduce the proposed coded beam training framework in RIS systems. In Section IV, we first introduce the proposed relaxed GS-based codeword design scheme, followed by the proposed dimension reduced encoder design scheme. Then, the necessary beam training overheads for the proposed framework and traditional frameworks are analyzed. Simulation results are provided in Section V, and conclusions are finally drawn in Section VI.

\subsubsection*{Notation}
Lower-case and upper-case boldface letters represent vectors and matrices, respectively; $\mathbf{v}\left(i\right)$ denotes the $i$-th element of the vector $\mathbf{v}$; $\mathbf{X}\left(i, j\right)$ denotes the $(i, j)$-th element of the matrix $\mathbf{X}$; $\mathbf{X}\left(i, :\right)$ and $\mathbf{X}\left(:, j\right)$ denote the $i$-th row and the $j$-th column of the matrix $\mathbf{X}$; $(\cdot)^T$ and $(\cdot)^H$ denote the transpose and conjugate transpose, respectively; $\left|\cdot\right|$ denotes the absolute operator; $\left\lVert\cdot\right\rVert _2$ denotes the $l_2$ norm operator; $\lceil \cdot \rceil $ denotes the ceiling operator; $\mathrm{mod}(\cdot)$ denotes the modulo operator; $\mathcal{CN}(\mu, \Sigma)$ and $\mathcal{U}(a, b)$ denote the Gaussian distribution with mean $\mu$ and covariance $\Sigma$, and the uniform distribution between $a$ and $b$, respectively.

\section{System Model and Background} 

\ifx\onecol\undefined
	\begin{figure*}[t!]
		\centering
		\includegraphics[width=1\linewidth]{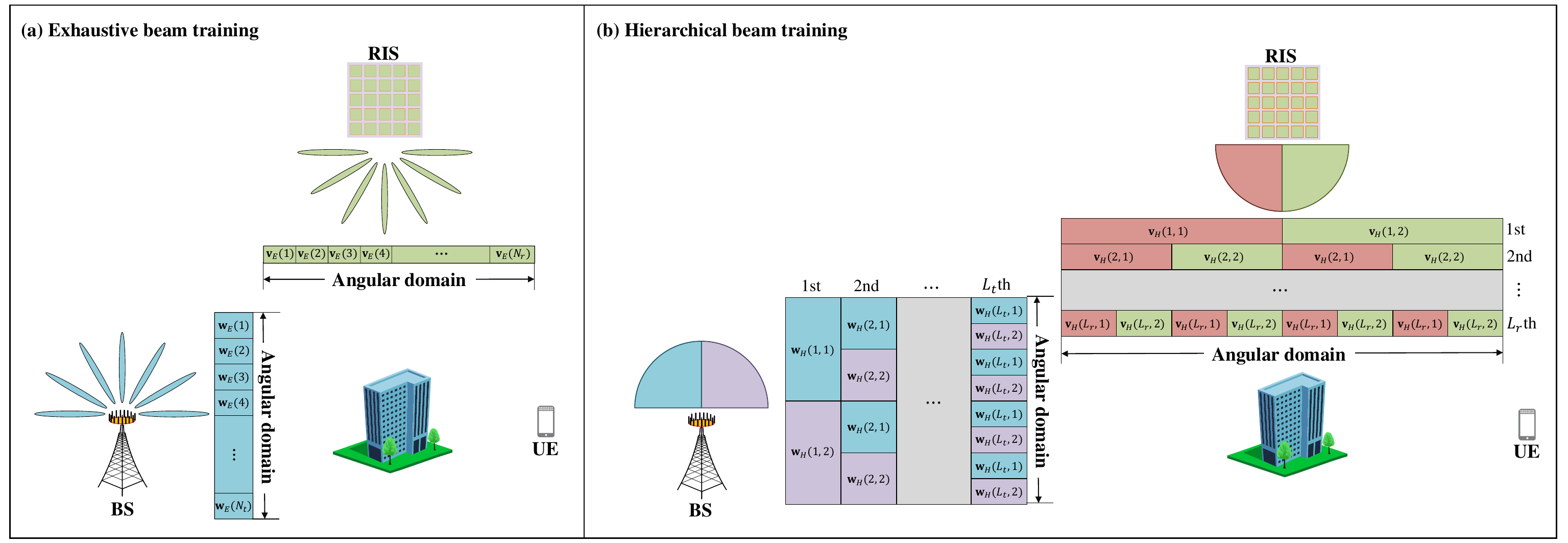}
		\caption{Traditional beam training frameworks. (a) Exhaustive beam training; (b) Hierarchical beam training.}
		\label{fig:tradition}
	\end{figure*}
\fi

In this section, we first introduce the system model of RIS assisted communication systems. Then, the traditional exhaustive and hierarchical beam training frameworks are reviewed.

\subsection{System Model}

We consider a downlink time division duplexing (TDD) RIS assisted communication system in this paper. The BS employs a uniform linear array (ULA) with $N_t$ antennas and the RIS employs a UPA with $N_{r_1}\times N_{r_2}=N_r$ antennas. The UE is equipped with a single antenna. We assume that the direct links between the BS and the UE are blocked by obstacles such as trees or buildings~\cite{wei2021channel}. Then, the received signal $y_p\in\mathbb{C}$ in the $p$-th time slot at the UE can be represented as 
\begin{equation}
	y_p=\mathbf{h}_r\mathrm{diag}(\mathbf{v}_p)\mathbf{G}\mathbf{w}_ps_p+n_p,
	\label{eq:yp_signal}
\end{equation}
where $\mathbf{h}_r\in\mathbb{C}^{1\times N_r}$ denotes the channel between UE and RIS; $\mathbf{v}_p\in\mathbb{C}^{N_r\times 1}$ denotes the reflecting vector of RIS at the $p$-th time slot; $\mathbf{G}\in\mathbb{C}^{N_r\times N_t}$ denotes the channel between RIS and BS; $\mathbf{w}_p$ denotes the beamforming vector of BS at the $p$-th time slot; $s_p\in\mathbb{C}$ denotes the signal sent by BS at the $p$-th time slot; and $n_p\sim\mathcal{CN}(0, \sigma^2)$ denotes the additive white Gaussian complex noise at the $p$-th time slot with $\sigma^2$ being the noise power, respectively. Due to the constant modulus constraint, RIS can only adjust the phase shift rather than the amplitude coefficient~\cite{dai2020reconfigurable}. As a result, the reflecting vector of the RIS can be re-written as $\mathbf{v}_p=\left[e^{j\vartheta_1}, e^{j\vartheta_2}, \cdots, e^{j\vartheta_{N_r}}\right]$, where $\vartheta_n\in\left[0, 2\pi\right], n=1, 2, \cdots, N_r$ represents the phase shift of the $n$-th element.

For the channel model, we apply a geometric channel model~\cite{alkhateeb2014channel}, which assumes each scatterer only contributes a single propagation path, so the channel $\mathbf{h}_r$ between UE and RIS can be written as 
\begin{equation}
	\mathbf{h}_r=\sqrt{\frac{N_r}{L_r}}\sum_{\ell=1}^{L_r}\alpha_\ell^r\mathbf{a}^T(\phi_\ell^r, \theta_\ell^r),
	\label{eq:hr_channel}
\end{equation}
where $L_r$ denotes the number of paths between UE and RIS; $\alpha_\ell^r=\frac{\lambda}{4\pi R_\ell^r}$ denotes the path gain of the $\ell$-th path with $R_\ell^r$ being the   distance of the $\ell$-th path; $\phi_\ell^r, \theta_\ell^r$ denote the azimuth angle and the elevation angle at RIS, respectively. The steering vector $\mathbf{a}(\phi, \theta)$ for a $N=N_1\times N_2$-antenna UPA can be elaborated as 
\begin{equation}
	\mathbf{a}(\phi, \theta)=\frac{1}{\sqrt{N}}\left[e^{-j2\pi d\sin(\phi)\sin(\theta)\bm{\delta}/\lambda}\right]\otimes\left[e^{-j2\pi d\cos(\theta)\bm{\varsigma}\lambda}\right],
	\label{eq:a_vector}
\end{equation}
where $\lambda=c/f_c$ denotes the wavelength of electromagnetic wave with $f_c$ being the central frequency and $c$ being the speed of light. The antenna spacing $d$ is set to $d=\lambda/2$. The antenna indices $\bm{\delta}$ and $\bm{\varsigma}$ can be represented as 
\begin{equation}
	\begin{aligned}
		\bm{\delta}=\left[\delta_1, \delta_2, \cdots, \delta_{N_1}\right]^T=\left[\tfrac{1-N_1}{2}, \tfrac{3-N_1}{2}, \cdots, \tfrac{N_1-1}{2}\right]^T\\
		\bm{\varsigma}=\left[\varsigma_1, \varsigma_2, \cdots, \varsigma_{N_2}\right]^T=\left[\tfrac{1-N_2}{2}, \tfrac{3-N_2}{2}, \cdots, \tfrac{N_2-1}{2}\right]^T
	\end{aligned}.
	\label{eq:index}
\end{equation}
Simularly, the channel $\mathbf{G}$ between RIS and BS can be represented as 
\begin{equation}
	\mathbf{G}=\sqrt{\frac{N_tN_r}{L_G}}\sum_{\ell=1}^{L_G}\alpha_\ell^G\mathbf{a}(\phi_\ell^{G_r}, \theta_\ell^{G_r})\mathbf{b}^T(\phi_\ell^{G_t}),
	\label{eq:G_channel}
\end{equation}
where $L_G$ denotes the number of paths between RIS and BS; $\alpha_\ell^G=\frac{\lambda}{4\pi R_\ell^G}$ denotes the path gain of the $\ell$-th path with $R_\ell^G$ being the transmission distance of the $\ell$-th path; $\phi_\ell^{G_r}, \theta_\ell^{G_r}, \phi_\ell^{G_t}$ denote the azimuth angle at RIS, the elevation angle at RIS and the azimuth angle at BS, respectively. The steering vector $\mathbf{b}(\phi)$ for a $N$-antenna ULA can be elaborated as 
\begin{equation}
	\mathbf{b}(\phi)=\frac{1}{\sqrt{N}}\left[1, e^{-j2\pi d\sin(\phi)/\lambda}, \cdots, e^{-j2(N-1)\pi d\sin(\phi)/\lambda}\right]^T.
	\label{eq:b_vector}
\end{equation}
Due to the severe loss incurred by the scattering, high-frequency communication heavily relies on the line-of-sight (LoS) path~\cite{han2021hybrid}, so we set $L_r=L_G=1$ in this paper, and~\eqref{eq:hr_channel}~\eqref{eq:G_channel} can be further written as
\begin{equation}
	\mathbf{h}_r=\sqrt{N_r}\alpha^r\mathbf{a}^T(\phi^r, \theta^r),
\end{equation}
\begin{equation}
	\mathbf{G}=\sqrt{N_tN_r}\alpha^G\mathbf{a}(\phi^{G_r}, \theta^{G_r})\mathbf{b}^T(\phi^{G_t}).
\end{equation}
This means that the channel is determined by the angle at BS and the angle at RIS.

\subsection{Traditional Beam Training Frameworks}\label{sec:2.2}


To determine the angle at BS and the angle at RIS, beam training is usually applied. By generating directional beams to search all angles in space, the above angles can be obtained according to the beam tuple with the maximum received power. Here, we introduce two types of traditional beam training frameworks in RIS assisted communication systems: exhaustive beam training and hierarchical beam training.

\subsubsection{Exhaustive Beam Training}
One intuitive way to estimate the angles is to exhaustively search all possible angles in space. As illustrated in Fig.~\ref{fig:tradition}(a), both BS and RIS apply codewords in exhaustive codebooks to generate narrow beams and sequentially search all possible angles in space. The codewords at the BS side and the RIS side are denoted as $\mathbf{w}_E(i), i=1, 2, \cdots, N_t$ and $\mathbf{v}_E(j), j=1, 2, \cdots, N_r$, respectively. After receiving and recording received powers from all beam tuples, the angle at BS and the angle at RIS are estimated according to the beam tuple with the maximum received power. Since narrow beams are applied in the exhaustive beam training framework, the codebook size is equal to the number of antenna elements at both sides. In our considered scenario, where BS is equipped with $N_t$ antenna elements and RIS is equipped with $N_r$ antenna elements, the necessary beam training overhead is $N_tN_r$. In future communication systems, the antenna number at both BS and RIS tend to be very large, which means the exhaustive beam training framework will suffer from an unacceptable beam training overhead.

\subsubsection{Hierarchical Beam Training}
In order to reduce the beam training overhead, we can apply the idea of hierarchical beam training framework in RIS assisted communication systems~\cite{wang2023hierarchical}. As illustrated in Fig.~\ref{fig:tradition}(b), both BS and RIS apply binary search based codebooks, so each layer contains two codewords. We denote the $i$-th codeword in the $j$-th layer at the BS side and the RIS side as $\mathbf{w}_H(j, i)$ and $\mathbf{v}_H(j, i)$, respectively. According to the property of binary search, the numbers codebook layers at the BS side and the RIS side are $k_t=\log_2(N_t)$ and $k_r=\log_2(N_r)$, respectively. To gradually narrow down the possible range of UE, the beam patterns of codewords in higher layers possess higher angular resolutions compared to codewords in lower layers. 

During the beam training procedure, the codewords are transmitted layer by layer. Specifically, at the first layer, the BS and RIS sequentially transmits four beam tuples to UE, which can be listed as $\left\lbrace\mathbf{w}_H(1, 1), \mathbf{v}_H(1, 1)\right\rbrace$, $\left\lbrace\mathbf{w}_H(1, 1), \mathbf{v}_H(1, 2)\right\rbrace$, $\left\lbrace\mathbf{w}_H(1, 2), \mathbf{v}_H(1, 1)\right\rbrace$ and $\left\lbrace\mathbf{w}_H(1, 2), \mathbf{v}_H(1, 2)\right\rbrace$. We set $\mathbf{u}\in\left\lbrace\left\lbrace0, 0\right\rbrace, \left\lbrace0, 1\right\rbrace, \left\lbrace1, 0\right\rbrace, \left\lbrace1, 1\right\rbrace\right\rbrace^{\max\left\lbrace\log_2(N_t), \log_2(N_r)\right\rbrace}$ as the tuple vector which describes the angle at BS and the angle at RIS. Then, we set $\mathbf{u}(1)=\left\lbrace0, 0\right\rbrace$ if the received power of beam tuple $\left\lbrace\mathbf{w}_H(1, 1), \mathbf{v}_H(1, 1)\right\rbrace$ is the maximum, $\mathbf{u}(1)=\left\lbrace0, 1\right\rbrace$ if the received power of beam tuple $\left\lbrace\mathbf{w}_H(1, 1), \mathbf{v}_H(1, 2)\right\rbrace$ is the maximum, $\mathbf{u}(1)=\left\lbrace1, 0\right\rbrace$ if the received power of beam tuple $\left\lbrace\mathbf{w}_H(1, 2), \mathbf{v}_H(1, 1)\right\rbrace$ is the maximum and $\mathbf{u}(1)=\left\lbrace1, 1\right\rbrace$ if the received power of beam tuple $\left\lbrace\mathbf{w}_H(1, 2), \mathbf{v}_H(1, 2)\right\rbrace$ is the maximum. After transmitting the beam tuples in all layers, the BS can decide the angle at BS and the angle at RIS based on the tuple vector $\mathbf{u}$. We take the first bit of each element in $\mathbf{u}$ as $\mathbf{u}_t$ and the second bit of each element in $\mathbf{u}$ as $\mathbf{u}_r$. The indices of the angle at BS and the angle at RIS can then be derived as $\mathrm{bin2dec}(\mathbf{u}_t)$ and $\mathrm{bin2dec}(\mathbf{u}_r)$, where $\mathrm{bin2dec}(\cdot)$ denotes the operation of transforming a binary number to a decimal number.

By searching the entire space layer by layer, we can exclude many incorrect angles without searching them in an exhaustive way and thus greatly improve the beam training efficiency. The overall beam training overhead is $4\max\left\lbrace k_t, k_r\right\rbrace=4\max\left\lbrace\log_2(N_t), \log_2(N_r)\right\rbrace$, which is far less than $N_tN_r$. However, during the binary hierarchical beam training, we need to generate wider beams than those in exhaustive beam training frameworks, so the beam gains are much smaller than those of narrow beams. In addition, in RIS assisted communication systems, there exists multiplicative fading effect~\cite{zhang2022active,zhi2022active}, which means that the equivalent path loss of the BS-RIS-UE link is the product of the path loss of BS-RIS link and the path loss of RIS-UE link. These two factors lead to a reduced SNR at UE during beam training, and thus severely limit the beam training accuracy in RIS assisted communication systems.

\section{Coded Beam Training Framework in RIS Systems}\label{sec:3}

To enhance the ability to realize accurate beam training under poor SNR conditions, we design a coded beam training framework for RIS assisted communication systems, which is inspired by the coded beam training framework for MIMO in~\cite{zheng2024coded}. By applying the idea of channel coding and adding redundant beam training pilots, the accidental error caused by random noise during the beam training can be corrected without feedback. 

Different from the scenario considered in~\cite{zheng2024coded}, in RIS assisted communication systems, we need to estimate the angles both at BS and at RIS, so the best beam tuple, rather than the best beam, needs to be determined. Specifically, for BS, there are $N_t$ candidate angles, and we need $k_t=\log_2(N_t)$ information bits to determine the angle at BS. Similarly, RIS has $N_r$ candidate angles, and we need $k_r=\log_2(N_r)$ information bits\footnote{Here, both $N_t$ and $N_r$ are not necessarily in the form of $(2^n)$. When not so, $k_t=\left\lceil\log_2(N_t)\right\rceil$ and $k_r=\left\lceil\log_2(N_r)\right\rceil$, where $\left\lceil\cdot\right\rceil$ is the ceiling operation.} to determine the angle at RIS. Similar to the hierarchical beam training framework in Section~\ref{sec:2.2}, the information bits at BS and RIS are $\mathbf{u}_t$ and $\mathbf{u}_r$, respectively. To leverage the error correction capability of channel coding, we need to encode the effective information bits by mapping the information bits $\mathbf{u}_t\in\left\lbrace0, 1\right\rbrace^{k_t}$ and $\mathbf{u}_r\in\left\lbrace0, 1\right\rbrace^{k_r}$ to codewords $\mathbf{x}_t\in\left\lbrace0, 1\right\rbrace^{n_t}$ and $\mathbf{x}_r\in\left\lbrace0, 1\right\rbrace^{n_r}$, where $k_t\leq n_t, k_r\leq n_r$. We denote the encoding function at BS and RIS as $f_t$ and $f_r$. Then, we have $\mathbf{x}_t=f_t(\mathbf{u}_t)$ and $\mathbf{x}_r=f_r(\mathbf{u}_r)$. It is worth noting that since we do not know the angle of UE before the beam training. Here, the encoding function is conducted to all possible $\mathbf{u}_t$ and $\mathbf{u}_r$. For a BS with $k_t=\log_2(N_t)$ elements, we need to encode all possible $N_t$ bitstreams and get the corresponding $\mathbf{x}_t$ for the next step. This is the same for RIS.

After encoding the information bits, we need to build the connection between the codewords $\mathbf{x}_t$ and $\mathbf{x}_r$ and the ideal beam pattern in space during beam training. We denote the candidate angle list at BS as $\bm{\Omega}_t\in\mathbb{R}^{N_t}$, which can be expressed as 
\begin{equation}
	\bm{\Omega}_t(n)=\sin^{-1}\left(-\frac{N_t+1}{N_t}+\frac{2n}{N_t}\right), n=1, 2, \cdots, N_t,
	\label{eq:Omega-t}
\end{equation}
and the candidate angle list at RIS as $\bm{\Omega}_r\in\mathbb{R}^{N_t\times2}$, which can be expressed as 
\ifx\onecol\undefined
\begin{equation}
	\begin{aligned}
		&\bm{\Omega}_r(n, 1)=\sin^{-1}\left[\left(-\frac{N_{r_1}+1}{N_{r_1}}+\frac{2\lceil\frac{n}{N_{r_1}}\rceil}{N_{r_1}}\right)/\sin(\bm{\Omega}_r(n, 2))\right]\\
		&\bm{\Omega}_r(n, 2)=\cos^{-1}\left(\frac{1-N_{r_2}}{N_{r_2}}+\frac{2\mathrm{mod}\left(n, N_{r_2}\right)}{N_{r_2}}\right)\\
		&\hspace{60mm}n=1, 2, \cdots, N_r,
	\end{aligned}
	\label{eq:Omega-r}
\end{equation}
\else 
\begin{equation}
	\begin{aligned}
		&\bm{\Omega}_r(n, 1)=\sin^{-1}\left[\left(-\frac{N_{r_1}+1}{N_{r_1}}+\frac{2\lceil\frac{n}{N_{r_1}}\rceil}{N_{r_1}}\right)/\sin(\bm{\Omega}_r(n, 2))\right]\\
		&\bm{\Omega}_r(n, 2)=\cos^{-1}\left(\frac{1-N_{r_2}}{N_{r_2}}+\frac{2\mathrm{mod}\left(n, N_{r_2}\right)}{N_{r_2}}\right)
	\end{aligned}, \qquad n=1, 2, \cdots, N_r,
	\label{eq:Omega-r}
\end{equation}
\fi
where $\bm{\Omega}_r(:, 1)$ denotes the azimuth angles and $\bm{\Omega}_r(:, 2)$ denotes the elevation angles. We denote the ideal beam pattern at BS and at RIS as $\mathcal{V}_t\in\left\lbrace0, 1\right\rbrace^{n_t\times N_t}$ and $\mathcal{V}_r\in\left\lbrace0, 1\right\rbrace^{n_r\times N_r}$. In our proposed beam training framework, we want the generated beam to cover a certain set of angles in space. For these angles, we set the corresponding elements in $\mathcal{V}_t$ or $\mathcal{V}_r$ as 1. For the set of angles that we do not want to cover by a certain beam, we set the corresponding elements in $\mathcal{V}_r$ or $\mathcal{V}_r$ as 0. Therefore, $\mathcal{V}_t$ and $\mathcal{V}_r$ can be obtained by
\begin{equation}
	\begin{aligned}
		\mathcal{V}_t(:, i)=\mathbf{x}_t^{(i)}, i=1, 2, \cdots, N_t\\
		\mathcal{V}_r(:, j)=\mathbf{x}_r^{(j)}, j=1, 2, \cdots, N_r
	\end{aligned},
	\label{eq:pattern}
\end{equation}
where $\mathbf{x}_t^{(i)}=f_t(\mathbf{u}_t^{(i)})$ and $\mathbf{x}_r^{(j)}=f_r(\mathbf{u}_r^{(j)})$. Here, $\mathbf{u}_t^{(i)}$ and $\mathbf{u}_r^{(j)}$ denote the information bits of different angle indices, which can be expressed as $\mathbf{u}_t^{(i)}=\mathrm{dec2bin}(i, k_t)$ and $\mathbf{u}_r^{(j)}=\mathrm{dec2bin}(j, k_r)$, where $\mathrm{dec2bin}(\cdot, \kappa)$ denotes the operation of transforming a decimal number to a binary number of $\kappa$ bits.

Based on the ideal beam pattern, the beam training codebook can be designed. We denote the beam training codebook at BS and at RIS as $\mathcal{C}_t\in\mathbb{C}^{n_t\times N_t\times2}$ and $\mathcal{C}_r\in\mathbb{C}^{n_r\times N_r\times2}$. Each layer has two codewords, and the beams generated by these two codewords should cover the entire space without overlapping. Specifically, for $\mathcal{C}_t$, at the $i$-th layer, the first codeword $\mathcal{C}_t(i, :, 1)$ should cover the angles $\bm{\Omega}_t(\varrho)$, where $\mathcal{V}_t(i, \varrho)=1$. On the contrary, the second codeword $\mathcal{C}_t(i, :, 2)$ should cover the angles $\bm{\Omega}_t(\varrho)$, where $\mathcal{V}_t(i, \varrho)=0$. For $\mathcal{C}_r$, at the $j$-th layer, the first codeword $\mathcal{C}_r(j, :, 1)$ should cover the angles $\bm{\Omega}_r(\varepsilon, :)$, where $\mathcal{V}_r(j, \varepsilon)=1$, while the second codeword $\mathcal{C}_r(j, :, 2)$ should cover the angles $\bm{\Omega}_r(\varepsilon)$, where $\mathcal{V}_r(j, \varepsilon)=0$. Here, $\varrho$ and $\varepsilon$ represent the angles within the intended coverage range of a certain codeword at BS and RIS, respectively.

After designing the beam training codebook, we can start the training process. Similar to the hierarchical beam training framework, at each layer (i.e., the $i$-th layer), BS and RIS sequentially transmit four beam tuples to UE, which can be listed as $\left\lbrace\mathcal{C}_t(i, :, 1), \mathcal{C}_r(i, :, 1)\right\rbrace$, $\left\lbrace\mathcal{C}_t(i, :, 1), \mathcal{C}_r(i, :, 2)\right\rbrace$, $\left\lbrace\mathcal{C}_t(i, :, 2), \mathcal{C}_r(i, :, 1)\right\rbrace$ and $\left\lbrace\mathcal{C}_t(i, :, 2), \mathcal{C}_r(i, :, 2)\right\rbrace$. Based on the received power of these four beam tuples, we obtain the received tuple vector $\hat{\mathbf{x}}$ similar to the hierarchical beam training. By seperating each tuple in $\hat{\mathbf{x}}$ into two bits, the received codewords corresponding to BS and RIS can be denoted as $\hat{\mathbf{x}}_t$ and $\hat{\mathbf{x}}_r$. Finally, the decoding function $g_t$ and $g_r$ can be applied to recover the original information bits by $\hat{\mathbf{u}}_t=g_t(\hat{\mathbf{x}}_t)$ and $\hat{\mathbf{u}}_r=g_r(\hat{\mathbf{x}}_r)$. Since we introduce redundant bits through encoding, the error in $\hat{\mathbf{x}}$ can be corrected, thus improving the beam training accuracy under poor SNR conditions.

\section{Codeword Design for the Proposed Coded Beam Training Framework}

The above framework endows the RIS system with the self-correction ability during beam training and will potentially improve the beam training accuracy. To fully unleash the performance of the proposed framework, the accurate beam shape design is essential. 

For BS, the codeword design is straightforward~\cite{qi2020hierarchical} . We can generate the multi-mainlobe codeword with a weighted summation of several array response vectors as
\begin{equation}
	\mathcal{C}=\sum_{\phi_i\in\tilde{\bm{\Omega}}_t}e^{j\psi_i}\mathbf{b}(\phi_i),
\end{equation}
where $\tilde{\bm{\Omega}}_t$ denotes the set of angles that the codeword needs to cover, and the auxiliary phase off-set $\psi_i$ can help guarantee a high gain within the intended angle range, which can be elaborated as $\psi_i=i\pi(-1+\frac{1}{N_t}), i=1, 2, \cdots, \left|\tilde{\bm{\Omega}}_t\right|$. The principle behind this choice of auxiliary phase off-set is that we want to maximize the beamforming gain at the transition angles between two adjacent narrow beams.

However, for RIS, to design a codeword that evenly covers a set of randomly distributed angles in space is not easy due to the constant modulus constraint of RIS elements. Therefore, in this section, we propose a relaxed Gerchberg-Saxton-based codeword design scheme and a dimension reduced encoder design scheme to approach the desired beam pattern at RIS from two aspects.

\subsection{Proposed Relaxed Gerchberg-Saxton-based Codeword Design Scheme}

There have been several works concerning the codeword design on RIS. For example, in~\cite{you2020fast}, authors generated multi-directional beams by dividing the RIS into several sub-arrays and let each sub-array point to a certain direction. For this scheme, there is a trade-off between the coverage range and the angular resolution. The more sub-arrays are divided, the larger range can be covered, but the angular resolution would decrease since the number of elements on each sub-array is small. For our proposed framework, the required beams need to cover half of the entire angle with a high angular resolution. The scheme in~\cite{you2020fast} cannot realize such a flexible and accurate coverage.

It is worth noting that for the beam training problem, the amplitude of the designed codeword at different angles greatly affects the beam training accuracy, while the specific phase is not that important. This characteristic makes the codeword design problem similar to the phase retrieval problem in the field of digital holography imaging. Gerchberg-Saxton (GS) algorithm is widely used to solve phase retrieval problem~\cite{gerchberg1972practical,bucci1990intersection}. By iteratively imposing the two amplitude measurements in the object plane and diffraction pattern plane, the phase information of the image can be recovered. Following the idea of GS algorithm, authors in~\cite{lu2023hierarchical} studied a GS-based codeword design scheme at BS by applying the power normalization to one updating process to satisfy the power constraint of the codeword. 

\begin{figure}[t!]
	\centering
	\subfigure{
		\includegraphics[width=8cm]{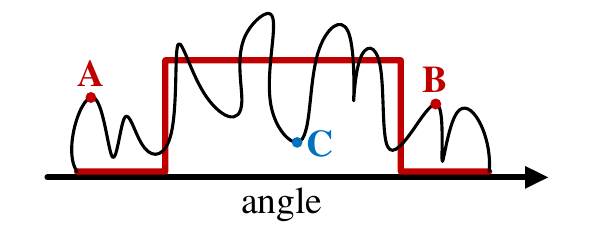}
		\label{fig:GS-origin}
	}\\

	(a)\\
	\subfigure{
		\includegraphics[width=8cm]{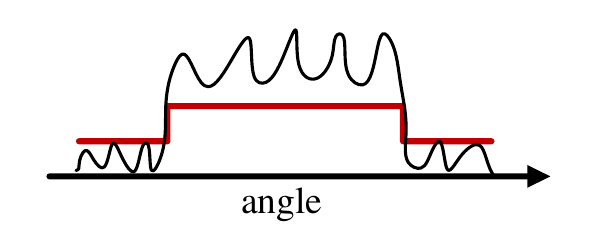}
		\label{fig:GS-relaxation}
	}\\

	(b)
	\caption{The designed beam shape when applying (a) GS-based codeword design scheme in~\cite{lu2023hierarchical}; (b) proposed relaxed GS-based codeword design scheme.}
	\label{fig:GS-compare}
\end{figure}

However, due to the constant modulus constraint at RIS, directly applying the codeword design scheme in~\cite{lu2023hierarchical} cannot yield an ideal beam shape and would possibly lead to some intrinsic errors even in high SNR scenarios. Specifically, as illustrated in Fig.~\ref{fig:GS-origin}, the red line represents the intended beam shape, the black line represents the generated beam shape, and the axis represents angles in space. We can see that due to the constant modulus constraint, the generated beam shape is not ideal and there exists oscillation. At angles out of the angle coverage range, we want the amplitude to be $0$ (i.e., point A and point B), but the amplitude is actually larger than some angles within the angle coverage range (i.e., point C). As a consequence, when UE is located in point A or point B, the received power of this codeword is actually high even though it should be near $0$. In this case, error chould happen even there is no noise in the system. The similar error also exists if we apply the scheme studied in~\cite{wang2023hierarchical}, where the beamforming configuration designed on RIS was modeled as a nonlinear optimization problem, and the problem was solved based on the Optimization Toolbox solver.

The root cause of this kind of error is that, the objective of the GS-based codeword design scheme in~\cite{lu2023hierarchical} is to minimize the difference between the intended beam shape and the generated beam shape, which can be formulated as 
\begin{equation}
	\begin{aligned}
		&\min\left\lVert\mathbf{A}^H\mathbf{v}-\mathbf{s}\right\rVert _2^2\\
		&\mathrm{s.t.} \quad \mathbf{v}(i)=e^{j\vartheta_i}, i=1, 2, \cdots, N_r,
	\end{aligned}
\end{equation}
where $\mathbf{A}\in\mathbb{C}^{N_r\times N_r}$ denotes the array response vectors at different angles, $\mathbf{v}\in\mathbb{C}^{N_r\times 1}$ denotes the designed codeword, and $\mathbf{s}\in\mathbb{C}^{N_r\times1}$ denotes the intended beam shape. Such an objective is reasonable in the realm of codeword design, but is not suitable for beam training. For beam training, the top priority is to distinguish between the angles within the angle coverage range and the angles out of the angle coverage range. To solve this problem. we propose a relaxed GS-based codeword design scheme. By relaxing the requirements of approaching the ideal beam pattern, we may not get the most similar beam shape, but we can distinguish the angles within the angle coverage range and the angles out of the angle coverage range more clearly. The procedure of the proposed relaxed GS-based codeword design scheme is elaborated in \textbf{Algorithm~\ref{alg:relax-GS}}.

\renewcommand{\algorithmicrequire}{\textbf{Input:}}
\renewcommand{\algorithmicensure}{\textbf{Output:}}
\begin{algorithm}[t!]
	\caption{Proposed relaxed GS-based codeword design scheme}
	\label{alg:relax-GS}
	\setstretch{1}
	\begin{algorithmic}[1]
		\REQUIRE $\tilde{\bm{\Omega}}_r$, $\mathbf{A}$, $K_{\mathrm{iter}}$, $\Delta$
		\ENSURE Designed codeword $\mathbf{v}$
		\STATE Initialize the intended beam shape $\mathbf{s}$ by~\eqref{eq:s-amp}
		\STATE Obtain the initial designed codeword $\hat{\mathbf{v}}_{(1)}$ by $\hat{\mathbf{v}}_{(1)}=\frac{1}{\sqrt{N_r}}e^{j\angle\left(\mathbf{A}^\dagger\mathbf{s}_{(0)}\right)}$
		\FOR {$k=1$ to $K_{\mathrm{iter}}$}
			\STATE Update $\mathbf{s}_{(k)}$ by~\eqref{eq:sk}
			\STATE Update $\Upsilon$ by~\eqref{eq:first}
			\STATE Update $\hat{\mathbf{s}}_{(k)}$ by~\eqref{eq:s-hat}
			\STATE Update $\hat{\mathbf{v}}_{(k+1)}$ by~\eqref{eq:v-hat}
		\ENDFOR
		\STATE Obtain the designed codeword $\mathbf{v}$ by $\mathbf{v}=\hat{\mathbf{v}}_{(K_{\mathrm{iter}}+1)}$
	\end{algorithmic}
\end{algorithm}

Here, $\tilde{\bm{\Omega}}_r$ denotes the angle coverage range of the intended beam shape, which is obtained based on $\mathcal{V}_r$. Matrix $\mathbf{A}$ transforms the codeword to the beam shape at the entire space, which can be expressed as 
\begin{equation}
	\mathbf{A}(:, n)=\mathbf{a}\left(\bm{\Omega}_r(n, 1), \bm{\Omega}_r(n, 2)\right).
\end{equation}
We denote the intended beam shape as 
\begin{equation}
	\mathbf{s}=\left[s(\phi_1, \theta_1), \cdots, s(\phi_1, \theta_{N_{r_2}}), \cdots, s(\phi_{N_{r_1}}, \theta_{N_{r_2}})\right],
	\label{eq:s-intend}
\end{equation} 
where $s(\phi, \theta)=\left|s(\phi, \theta)\right|e^{j\varphi(\phi, \theta)}$ with $\left|s(\phi, \theta)\right|$ being the intended amplitude and $\varphi(\phi, \theta)$ being the phase information. We hope angles within the angle coverage range can receive signals and angles out of the angle coverage range cannot receive signals, so the intended amplitude $\left|s(\phi, \theta)\right|$ should be 
\begin{equation}
	\left|s(\phi, \theta)\right|=\left\lbrace
	\begin{aligned}
		\mathcal{P} && \left(\phi, \theta\right)\in\tilde{\bm{\Omega}}_r\\
		0			&& \left(\phi, \theta\right)\notin\tilde{\bm{\Omega}}_r
	\end{aligned},
	\right.
	\label{eq:s-amp}
\end{equation}
where $\mathcal{P}$ is the constant decided by the codeword power. This equation means that we want the amplitudes of angles within the angle coverage range to be non-zero constants, and the amplitudes of angles out of the angle coverage range to be zero. At the beginning of the iteration, we initialize the intended beam shape $\mathbf{s}_{(0)}$ by~\eqref{eq:s-intend}, where the amplitude is generate by~\eqref{eq:s-amp} and the phase information $\varphi(\phi, \theta)$ is generated randomly. Based on $\mathbf{s}_{(0)}$, we initialize the designed codeword $\hat{\mathbf{v}}_{(1)}$ as $\hat{\mathbf{v}}_{(1)}=\frac{1}{\sqrt{N_r}}e^{j\angle\left(\mathbf{A}^\dagger\mathbf{s}_{(0)}\right)}$, where $\angle(\cdot)$ denotes the phase operator.

In the $k$-th round of iteration, we first calculate the beam shape realized by the designed codeword $\hat{\mathbf{v}}_{(k)}$ as 
\begin{equation}
	\mathbf{s}_{(k)}=\mathbf{A}^H\hat{\mathbf{v}}_{(k)}.
	\label{eq:sk}
\end{equation}
Here, the matrix $\mathbf{A}$ converts the codeword to the amplitudes corresponding to different angles. Then, different from the scheme in~\cite{lu2023hierarchical}, where the intended amplitude $\left|s(\phi, \theta)\right|$ is directly assigned to $\mathbf{s}_{(k)}$, we divide the points in $\mathbf{s}_{(k)}$ into two categories. For the first category, the corresponding amplitude can distinguish the angles within the angle coverage range from the angles out of the angle coverage range. The set of the points in the first category can be expressed as  
\begin{equation}
	\begin{aligned}
	\Upsilon=&\left\lbrace(\phi, \theta)\mid\left((\phi, \theta)\in\tilde{\bm{\Omega}}_r\:\&\&\:\mathbf{s}_{(k)}(\phi, \theta)\geq \mathcal{P}(1-\Delta)\right)\right.\\
	&\hspace{15mm}\left.||\left((\phi, \theta)\notin\tilde{\bm{\Omega}}_r\:\&\&\:\mathbf{s}_{(k)}(\phi, \theta)\leq \mathcal{P}\Delta\right)\right\rbrace,
	\end{aligned}
	\label{eq:first}
\end{equation}
where $\Delta\in\left[0, 0.5\right]$ is the dividing factor. For points in the first category, the amplitudes in this round have already revealed the difference between angles within the angle coverage range and angles out of the angle coverage range, so we relax the requirements on them by not assigning the exact amplitude in $\mathbf{s}$ to them. On the contrary, for the second category, where $(\phi, \theta)\notin\Upsilon$, we assign new amplitude to them, so $\hat{\mathbf{s}}_{(k)}$ can be expressed as 
\begin{equation}
	\begin{aligned}
		&\hspace{-2mm}\hat{\mathbf{s}}_{(k)}(\phi, \theta)=\\
		&\hspace{-1mm}\left\lbrace
		\begin{aligned}
			&\mathbf{s}_{(k)}(\phi, \theta)  &&(\phi, \theta)\in\Upsilon\\
			&\mathcal{P}(1-\Delta)e^{j\angle(\mathbf{s}_{(k)}(\phi, \theta))} && (\phi, \theta)\notin\Upsilon \:\&\&\: (\phi, \theta)\in\tilde{\bm{\Omega}}_r\\
			&\mathcal{P}\Delta e^{j\angle(\mathbf{s}_{(k)}(\phi, \theta))} && (\phi, \theta)\notin\Upsilon \:\&\&\: (\phi, \theta)\notin\tilde{\bm{\Omega}}_r
		\end{aligned}
		\right.
	\end{aligned}
	\label{eq:s-hat}
\end{equation}
This equation means for the first category, we do not assign a specific value to the amplitude. For the second category, we assign a value that can exactly make them distinguish from each other.
Based on $\hat{\mathbf{s}}_{(k)}$, the designed codeword for the next round of iteration can be obtained by 
\begin{equation}
	\hat{\mathbf{v}}_{(k+1)}=\frac{1}{\sqrt{N_r}}e^{j\angle\left(\mathbf{A}^\dagger\hat{\mathbf{s}}_{(k)}\right)}.
	\label{eq:v-hat}
\end{equation}
After $K_{\mathrm{iter}}$ rounds of iteration, the designed codeword $\mathbf{v}$ can finally be obtained as $\mathbf{v}=\hat{\mathbf{v}}_{(K_{\mathrm{iter}}+1)}$.

With the proposed relaxed GS-based codeword design scheme, we can generate beam shape like Fig.~\ref{fig:GS-relaxation}, where the angles within the angle coverage range and the angles out of the angle coverage range can be clearly distinguished.

\subsection{Proposed Dimension Reduced Encoder Design Scheme}\label{sec:4.2}

\begin{figure}[t!]
	\centering
	\subfigure{
		\includegraphics[width=4.1cm]{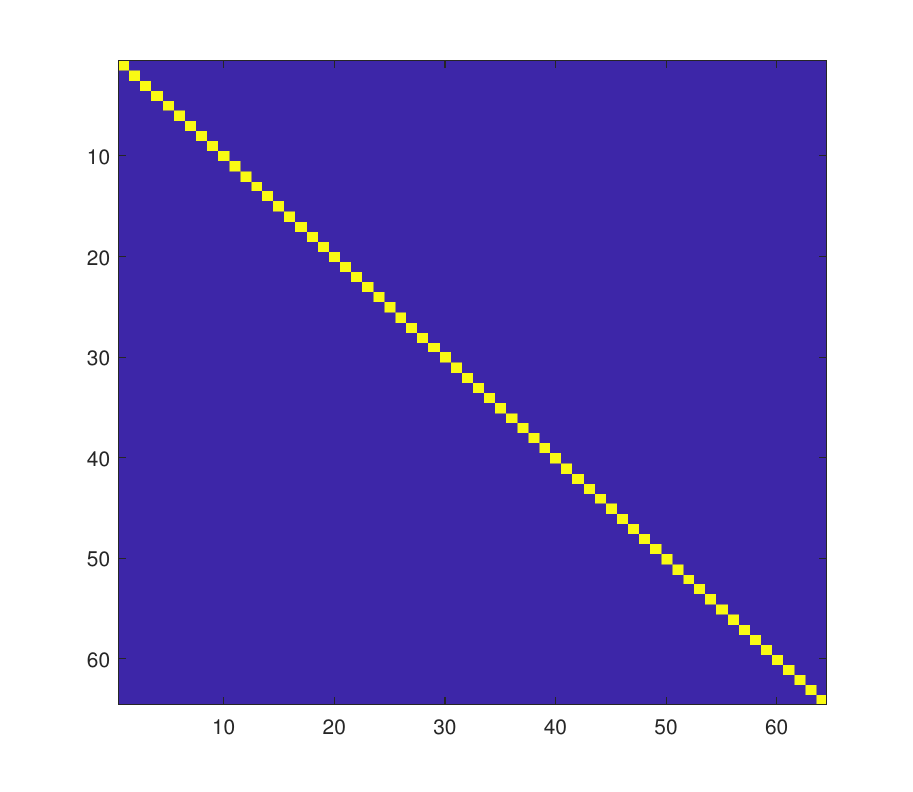}
		\label{fig:1D}
	}
	\subfigure{
		\includegraphics[width=4.1cm]{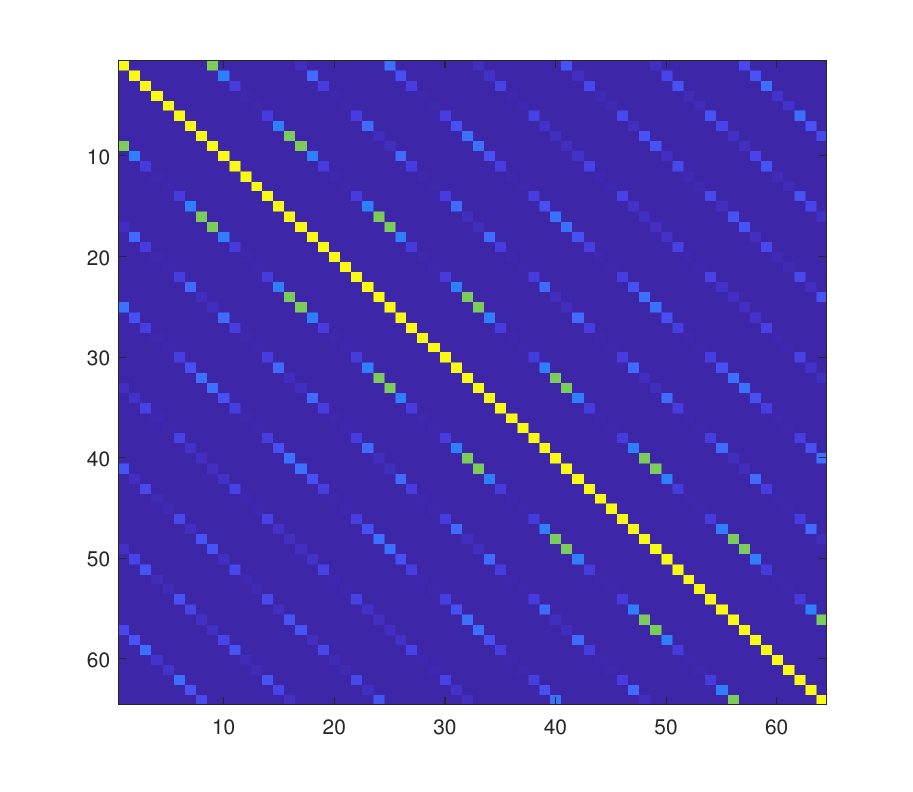}
		\label{fig:2D}
	}\\

	\quad(a)\qquad\qquad\qquad\qquad\qquad\qquad(b)
	\caption{The orthogonality of $\mathbf{A}$ when the size of RIS is (a) $64\times1$; (b) $8\times8$.}
	\label{fig:A-orthogonal}
\end{figure}

The above scheme works well for RIS with ULA, but for a RIS with UPA, the generated beam shape is still highly non-ideal. This is because the orthogonality of $\mathbf{A}$ in the two dimensional case is bad. As illustrated in Fig.~\ref{fig:A-orthogonal}, for the case where RIS is equipped with a $64\times1$ ULA, the orthogonality of $\mathbf{A}$ is good, while for case were RIS is equipped with a $8\times8$ UPA, the orthogonality of $\mathbf{A}$ is bad. Given the fact that RIS usually possesses a UPA structure to guarantee enough reflection area, how to generate good beam shape for a RIS with UPA is crucial to ensure a high beam training accuracy. Since the above scheme works well for RIS with ULA, can we decouple the 2D beam shape design problem into two 1D beam shape design problems? To enable this decoupling, the intended beam shape in space should be independent in the two dimensions. As discussed in Section~\ref{sec:3}, the intended beam shape is determined by the encoding function $f_r(\cdot)$. Therefore, we propose a dimension reduced encoder design scheme to decouple the 2D beam shape design problem into two 1D beam shape design problem so as to improve the beam training accuracy.

\begin{figure}[t!]
	\centering
	\includegraphics[width=0.9\linewidth]{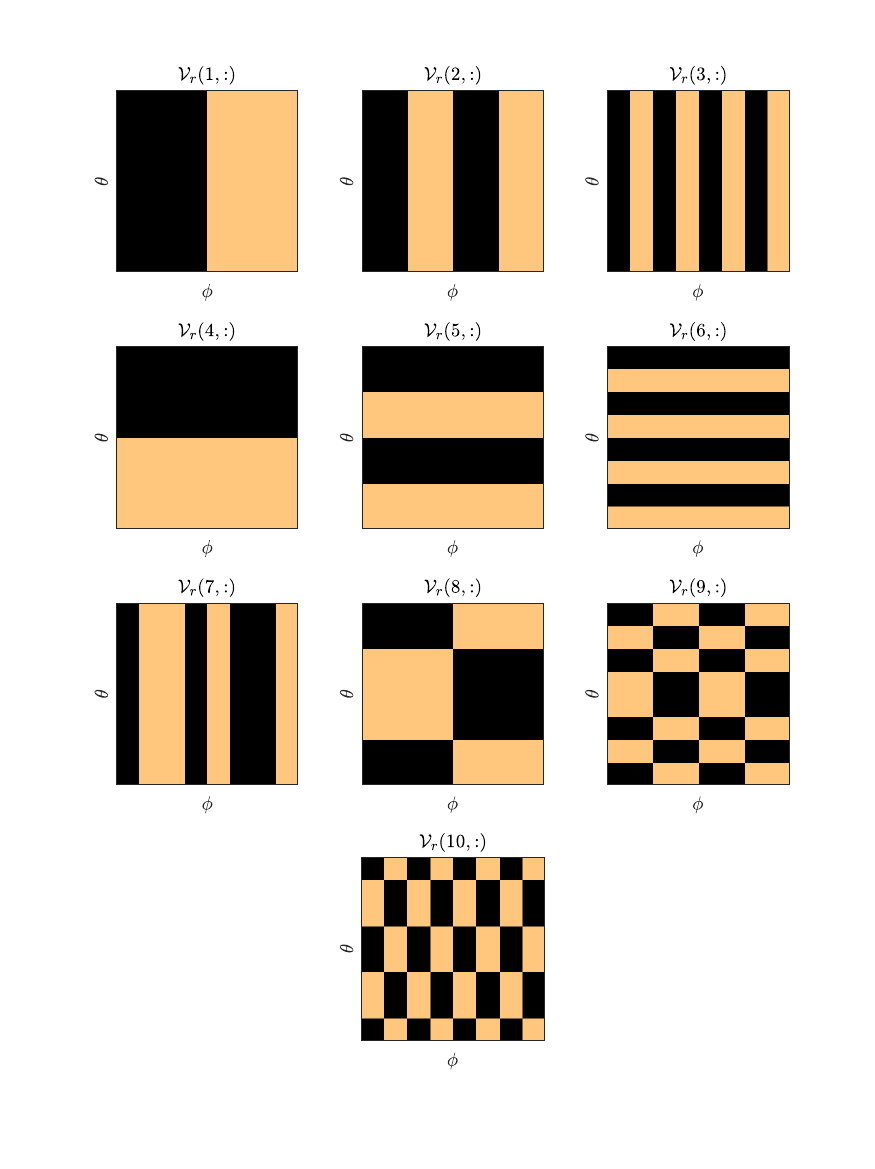}
	\vspace{-35pt}
	\caption{Angular coverage range corresponding to the ideal beam pattern $\mathcal{V}_r$.}
	\label{fig:G-origin}
\end{figure}


We choose the Hamming code as the coding scheme for our proposed coded beam training framework because of the high degree of freedom of Hamming code regarding the encoder design. Specifically, a $(n_r, k_r)$ Hamming code can encode a $k_r$-bit bitstream $\mathbf{u}_r^{(i)}$ into a $n_r$-bit codeword $\mathbf{x}_r^{(i)}$ with a generator matrix $\mathbf{G}_{\mathrm{Ham}}\in\left\lbrace0, 1\right\rbrace^{k_r\times n_r}$ by $\mathbf{x}_r^{(i)}=f_r(\mathbf{u}_r^{(i)})=\mathbf{u}_r^{(i)}\mathbf{G}_{\mathrm{Ham}}$. The generator matrix $\mathbf{G}_{\mathrm{Ham}}$ has the structure 
\begin{equation}
	\mathbf{G}_{\mathrm{Ham}}=
	\begin{bmatrix}
		\mathbf{I}_{k_r} & \mathbf{Q}	
	\end{bmatrix},
\end{equation}
where $\mathbf{I}_{k_r}$ denotes the $k_r\times k_r$ identical matrix. Submatrix $\mathbf{Q}\in\left\lbrace0, 1\right\rbrace^{k_r\times(n_r-k_r)}$ is designed artificially. To guarantee the error correction ability, each row of $\mathbf{Q}$ should contain at least two “1”. We consider a $8\times8$ RIS, the necessary number of information bits $k_r=\log_2(8\times8)=6$. The codeword length $n_r$ should thus be at least $n_r=10$~\cite{lin2004error}\footnote{Here, the Hamming code with 6 information bits and 4 redundant bits is not the classical Hamming code with the form (7, 4), (15,11), (31,26), etc. We do this extension since in practical systems, the antenna number may not exactly correspond to the number of information bits in classical Hamming code.}. We first randomly generate $\mathbf{Q}$ as 
\begin{equation}
	\mathbf{Q}=\begin{bmatrix}
		1 & 1 & 1 & 0 & 0 & 0\\
		1 & 0 & 0 & 1 & 1 & 0\\
		0 & 1 & 0 & 1 & 0 & 1\\
		0 & 0 & 1 & 0 & 1 & 1
	\end{bmatrix}^T.
\end{equation}
Based on the matrix $\mathbf{G_{\mathrm{Ham}}}$, we map the ideal beam pattern to the corresponding angular coverage range. The $n$-th column in $\mathcal{V}_r(1, :)$ is corresponding to the elevation angle $\theta=\cos^{-1}\left(\frac{1-N_{r_2}}{N_{r_2}}+\frac{2\mathrm{mod}\left(n, N_{r_2}\right)}{N_{r_2}}\right)$ and the azimuth angle $\phi=\sin^{-1}\left[\left(-\frac{N_{r_1}+1}{N_{r_1}}+\frac{2\lceil\frac{n}{N_{r_1}}\rceil}{N_{r_1}}\right)/\sin(\theta)\right]$, which is elaborated in~\eqref{eq:Omega-r}. When doing the mapping, unlike filling the ideal beam pattern matrix, where the encoded bitstreams are filled by columns, we do the mapping by rows to make each codeword cover half angles in the entire space.
The mapping results are depicted in Fig.~\ref{fig:G-origin}. 
The first 6 beam patterns are the same as those in hierarchical beam training frameworks. For these beam patterns, the two dimension (i.e., $\phi$ and $\theta$) can be decoupled. For example, to generate beam pattern $\mathcal{V}_r(1, :)$, we can design a 1D beam $\mathbf{v}_{\phi}\in\mathbb{C}^{8\times1}$ that covers $\phi\in\left[0, \pi/2\right]$ and a 1D beam $\mathbf{v}_{\theta}\in\mathbb{C}^{8\times1}$ that covers $\theta\in\left[0, \pi\right]$. The 2D beam can be realized by codeword $\mathbf{v}=\mathbf{v}_{\phi}\otimes\mathbf{v}_{\theta}$. However, for the last 4 beam patterns (the redundant beam patterns for error correction), only the $7^{\mathrm{th}}$ beam pattern can be decoupled into two 1D beams, and the other 3 beam patterns cannot be decoupled since the $\phi$-axis and the $\theta$-axis are interwoven with each other. 

What leads to this interweave? Since the first 6 columns of $\mathbf{G}_{\mathrm{Ham}}$ is an identity matrix, we can actually view the first 6 beam patterns as the \textbf{basis} patterns. Therefore, for the $7^{\mathrm{th}}$ beam pattern, according to the first column of $\mathbf{Q}$, it is obtained by adding up the first three basis (i.e., $\left[\mathcal{V}_r(1, :)+\mathcal{V}_r(2, :)+\mathcal{V}_r(3, :)\right]_2$), where $\left[\cdot\right]_2$ denotes the $\mathrm{mod}$-2 arithmetic. Since the first three beam patterns are all consistent at $\theta$-axis and varying at $\phi$-axis, so the beam pattern $\mathcal{V}_r(7, :)$ can still be decoupled. However, for the $8^{\mathrm{th}}$ beam pattern, it is obtained by $\left[\mathcal{V}_r(1, :)+\mathcal{V}_r(4, :)+\mathcal{V}_r(5, :)\right]_2$. Since $\mathcal{V}_r(4, :)$ and $\mathcal{V}_r(5, :)$ are consistent at $\phi$-axis and varying at $\theta$-axis, they will interweave with $\mathcal{V}_r(1, :)$ and make $\mathcal{V}_r(8, :)$ unable to be decoupled. Similarly, $\mathcal{V}_r(9, :)$ and $\mathcal{V}_r(10, :)$ are also unable to be decoupled. 

The above analysis inspires us that if we need to decouple the redundant beam patterns, we need to design the matrix $\mathbf{Q}$ so that only the basis with the same consistency is added together. For the simplicity of description, beam patterns that are consistent at $\theta$-axis and varying at $\phi$-axis are defined as \textbf{Type I} pattern, while beam patterns that are consistent at $\phi$-axis and varying at $\theta$-axis are defined as \textbf{Type II} pattern. To guarantee the error correction ability, we have the following \textbf{Proposition~\ref{lemma1}}.
\begin{proposition}
	To guarantee the error correction ability at RIS, the number of RIS elements at each dimension should be strictly larger than $4$ and the number of redundant beam patterns for each dimension should be at least $3$.\footnote{This condition holds for all linear block code. For other codes like convolutional code, this condition can be further relaxed. In practical systems, we should choose proper coding method according to the system parameters.}
	\label{lemma1}
\end{proposition}
\begin{proof}
	For a block code such as Hamming code, the error correction ability is related to its minimum hamming distance $d_{\mathrm{min}}$, which is equal to the minimum Hamming weight of its nonzero codewords~\cite{lin2004error}. In order to correct $t$ bits errors, $d_{\mathrm{min}}$ should satisfy $d_{\mathrm{min}}\geq 2t+1$. In our framework, we hope the Hamming code can correct $1$ bit error, so $d_{\mathrm{min}}\geq 3$. As a result, for each row in $\mathbf{G}_{\mathrm{Ham}}$, we should have at least three “1”, which means each row of $\mathbf{Q}$ should have at least two “1”. As discussed in Section~\ref{sec:4.2}, the two types of beam patterns cannot co-exist in the same column of $\mathbf{Q}$, so $\mathbf{Q}$ can be written as a block matrix as 
	\begin{equation}
		\mathbf{Q}=\begin{bmatrix}
			\mathbf{Q}_{\mathrm{I}} & \mathbf{0}\\
			\mathbf{0} & \mathbf{Q}_{\mathrm{II}}
		\end{bmatrix},
	\end{equation}
where $\mathbf{Q}_{\mathrm{I}}$ and $\mathbf{Q}_{\mathrm{II}}$ denote the submatrix related to \textbf{Type I} pattern and \textbf{Type II} pattern, respectively. Therefore, $\mathbf{Q}_{\mathrm{I}}$ and $\mathbf{Q}_{\mathrm{II}}$ should both have at least two “1”. Since \textbf{Type I} pattern and \textbf{Type II} pattern are homogeneous, we will only discuss \textbf{Type I} pattern and $\mathbf{Q}_{\mathrm{I}}$ in the following discussion. 

	If $N_{r_1}\leq 4$, $\mathbf{Q}_{\mathrm{I}}$ has at most two rows. To avoid repeated beam pattern, $\mathbf{Q}_{\mathrm{I}}$ would only be $\left[1 \quad 1\right]^T$. In this case, the Hamming weights of first two rows of $\mathbf{G}_{\mathrm{Ham}}$ are only $2$, which means $d_{\mathrm{min}}=2$, and the Hamming code can no longer correct 1 bit error. As a result, the number of RIS elements at each dimension should be strictly larger than $4$.

	Next, we need to prove that the number of columns in $\mathbf{Q}_{\mathrm{I}}$ should be at least $3$. If $\mathbf{Q}_{\mathrm{I}}$ only has two columns, in order for $d_{\mathrm{min}}\geq3$, all columns in $\mathbf{Q}_{\mathrm{I}}$ should be “1”. In this case, if we calculate the difference of these two rows in $\mathbf{G}_{\mathrm{Ham}}$, we can get a codeword with Hamming weight $d_{\mathrm{min}}=2$. As a result, the number of columns in $\mathbf{Q}_{\mathrm{I}}$ should be at least $3$. Similarly, the number of columns in $\mathbf{Q}_{\mathrm{II}}$ should also be at least $3$, which completes the proof.
\end{proof}

Based on the above analyses, we now introduce the steps of the proposed dimension reduced encoder design scheme. Since RIS is equipped with $N_r=N_{r_1}\times N_{r_2}$ antenna elements, we have $\log_2(N_{r_1})$ \textbf{Type I} patterns and $\log_2(N_{r_2})$ \textbf{Type II} patterns. We denote the number of redundant beam patterns for \textbf{Type I} patterns and \textbf{Type II} patterns as $m_{r_1}$ and $m_{r_2}$ respectively, then they should satisfy
\begin{equation}
	\left\lbrace
	\begin{aligned}
		m_{r_1} = \max\left\lbrace3, m_{r_1, \mathrm{int}}\right\rbrace\\
		m_{r_2} = \max\left\lbrace3, m_{r_2, \mathrm{int}}\right\rbrace
	\end{aligned}
	\right.,
\end{equation}
where $m_{r_1, \mathrm{int}}$ denotes the minimum integer that satisfies $2^{m_{r_1, \mathrm{int}}}-m_{r_1, \mathrm{int}}-1\geq \log_2(N_{r_1})$ and $m_{r_2, \mathrm{int}}$ denotes the minimum integer that satisfies $2^{m_{r_2, \mathrm{int}}}-m_{r_2, \mathrm{int}}-1\geq \log_2(N_{r_2})$~\cite{lin2004error}. For redundant beam patterns corresponding to \textbf{Type I} patterns, each row of $\mathbf{Q}_{\mathrm{I}}\in\left\lbrace0, 1\right\rbrace^{\log_2(N_{r_1})\times m_{r_1}}$ should be composed of $m_{r_1}$-tuples of weight 2 or more. There are a total of $\sum_{i=2}^{m_{r_1}}C(m_{r_1}, i)=2^{m_{r_1}}-m_{r_1}-1$ types of $m_{r_1}$-tuples of weight 2 or more, so we can always fill $\mathbf{Q}_{\mathrm{I}}$ without repeating existing tuples. Meanwhile, $\mathbf{Q}_{\mathrm{II}}\in\left\lbrace0, 1\right\rbrace^{\log_2(N_{r_2})\times m_{r_2}}$ can be generated by the same way. Finally, the submatrix $\mathbf{Q}$ can be composed by 
\begin{equation}
	\mathbf{Q}=\begin{bmatrix}
		\mathbf{Q}_{\mathrm{I}}	&	\mathbf{0}_{\log_2(N_{r_1})\times m_{r_2}}	\\
		\mathbf{0}_{\log_2(N_{r_2})\times m_{r_1}}	&	\mathbf{Q}_{\mathrm{II}}
	\end{bmatrix},
\end{equation}
where $\mathbf{0}_{\iota \times \gamma}$ denotes the all-zero matrix with dimension $\iota\times\gamma$.

\begin{figure}[t!]
	\centering
	\includegraphics[width=0.9\linewidth]{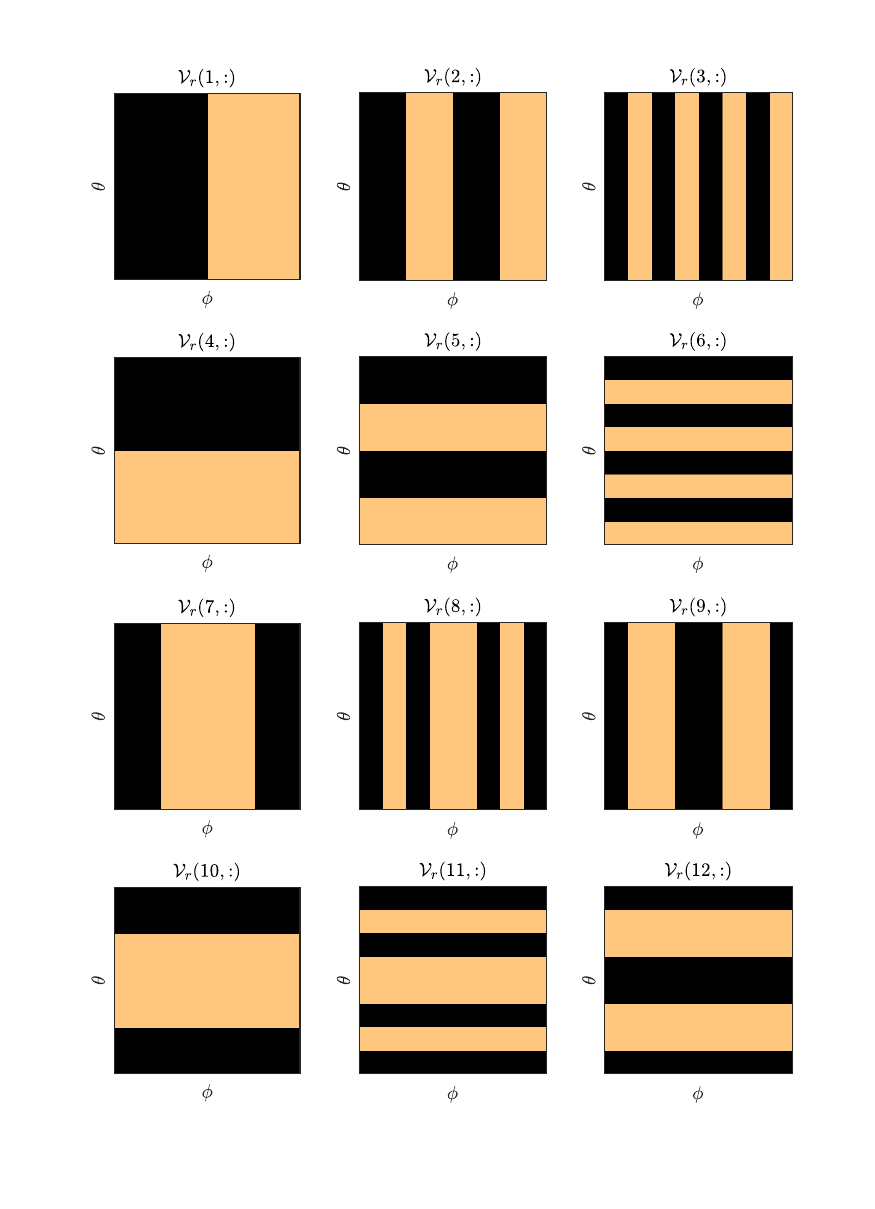}
	\vspace{-35pt}
	\caption{Angular coverage range corresponding to the proposed encoder.}
	\label{fig:G-reduced}
\end{figure}

With the proposed scheme, now we get back to the above example where the RIS is equipped with $8\times8$ elements. In this case, $m_{r_1}=m_{r_2}=3$, and submatrix $\mathbf{Q}$ can be generated as 
\begin{equation}
	\mathbf{Q}=\left[\begin{array}{ccc:ccc}
		1 & 1 & 0 & 0 & 0 & 0 \\
		1 & 0 & 1 & 0 & 0 & 0 \\
		0 & 1 & 1 & 0 & 0 & 0 \\
		\hdashline
		0 & 0 & 0 & 1 & 1 & 0 \\
		0 & 0 & 0 & 1 & 0 & 1 \\
		0 & 0 & 0 & 0 & 1 & 1 \\
	\end{array}\right].
\end{equation}
Then, the $n_r=k_r+m_{r_1}+m_{r_2}$ beam patterns $\mathcal{V}_r$ are depicted in Fig.~\ref{fig:G-reduced}. Through the proposed dimension reduced encoder design scheme, although we need two more redundant beam patterns, the two dimensions of RIS are properly decoupled and the quality of generated beam shape can be guarenteed. 

\subsection{The Decoding Procedure}

Based on the designed $\mathbf{Q}$, the check matrix $\mathbf{H}$ can be expressed as 
\begin{equation}
	\mathbf{H}=\begin{bmatrix}
		\mathbf{Q}^T & \mathbf{I}_{m_{r_1}+m_{r_2}}
	\end{bmatrix},
\end{equation}
based on which we can determine whether the received codeword $\hat{\mathbf{x}}_r$ contains error bits by calculating the syndrome $\mathbf{c}_r$ as 
\begin{equation}
	\mathbf{c}_r=\hat{\mathbf{x}}_r\mathbf{H}^T.
\end{equation}
When all bits in $\mathbf{c}_r$ equal to zero, the received $\hat{\mathbf{x}}_r$ is a normal codeword and there is no error. On the contrary, when $\mathbf{c}_r\neq \mathbf{0}_{1\times (m_{r_1}+m_{r_2})}$, $\hat{\mathbf{x}}_r$ is not a normal codeword generated by the designed codebook and there exists error in $\hat{\mathbf{x}}_r$. For the considered case here, when the first bit of $\hat{\mathbf{x}}_r$ is wrong, $\mathbf{c}_r=\left[1, 1, 0, 0, 0, 0\right]$. In other words, when $\mathbf{c}_r=\left[1, 1, 0, 0, 0, 0\right]$, we know the first bit of $\hat{\mathbf{x}}_r$ is wrong. After flipping the wrong bit, the bitstream $\hat{\mathbf{u}}_r$ can be estimated by extracting the first $k_r$ bits of $\hat{\mathbf{x}}_r$. Since $d_{\mathrm{min}}=3$, all 1-bit error has a unique syndrome and can be corrected. In addition, when there exists 1-bit error in \textbf{Type I} pattern and 1-bit error in \textbf{Type II} pattern simultaneously, these 2-bit errors can also be corrected according to the following \textbf{Proposition~\ref{lemma2}}.
\begin{proposition}
	By the proposed dimension reduced encoder design scheme, the error in \textbf{Type I} patterns and the error in \textbf{Type II} patterns are independent with each other. 
	\label{lemma2}
\end{proposition}
\begin{proof}
	For a certain information bit stream $\mathbf{u}_r^{(j)}$, it can be divided into two parts: the bits corresponding to the $\phi$ dimension with length $\log_2(N_{r_1})$, denoted as $\mathbf{u}_{r, \mathrm{I}}^{(i)}$, and the bits corresponding to the $\theta$ dimension with length $\log_2(N_{r_2})$, denoted as $\mathbf{u}_{r, \mathrm{II}}^{(i)}$. The generated codeword $\mathbf{x}_r^{(i)}$ can then be derived as 
	\begin{equation}
		\begin{aligned}
			\mathbf{x}_r^{(i)}&=\begin{bmatrix}
				\mathbf{u}_{r, \mathrm{I}}^{(i)} & \mathbf{u}_{r, \mathrm{II}}^{(i)}
			\end{bmatrix}\begin{bmatrix}
				\mathbf{I}_{\log_2(N_{r_1})} & \mathbf{0} & \mathbf{Q}_{\mathrm{I}} & \mathbf{0} \\
				\mathbf{0} & \mathbf{I}_{\log_2(N_{r_2})} & \mathbf{0} & \mathbf{Q}_{\mathrm{II}}
			\end{bmatrix}\\
			&=\begin{bmatrix}
				\mathbf{u}_{r, \mathrm{I}}^{(i)} & \mathbf{u}_{r, \mathrm{II}}^{(i)} & \mathbf{u}_{r, \mathrm{I}}^{(i)}\mathbf{Q}_{\mathrm{I}} & \mathbf{u}_{r, \mathrm{II}}^{(i)}\mathbf{Q}_{\mathrm{II}}
			\end{bmatrix}.
		\end{aligned}
	\end{equation}
After the transmission, we denote the received codeword as 
\begin{equation}
	\hat{\mathbf{x}}_r^{(i)}=\begin{bmatrix}
		\hat{\mathbf{u}_{r, \mathrm{I}}^{(i)}} & \hat{\mathbf{u}_{r, \mathrm{II}}^{(i)}} & \hat{\mathbf{u}_{r, \mathrm{I}}^{(i)}\mathbf{Q}_{\mathrm{I}}} & \hat{\mathbf{u}_{r, \mathrm{II}}^{(i)}\mathbf{Q}_{\mathrm{II}}}
	\end{bmatrix}.
\end{equation}
Then, the syndrome $\mathbf{c}_r$ can be derived as 
\begin{equation}
	\begin{aligned}
	\mathbf{c}_r&=\hat{\mathbf{x}}_r^{(i)}\mathbf{H}^T=\begin{bmatrix}
		\mathbf{c}_{r, \mathrm{I}} & \mathbf{c}_{r, \mathrm{II}}
	\end{bmatrix}\\
	&=\begin{bmatrix}
		\hat{\mathbf{u}_{r, \mathrm{I}}^{(i)}} & \hat{\mathbf{u}_{r, \mathrm{II}}^{(i)}} & \hat{\mathbf{u}_{r, \mathrm{I}}^{(i)}\mathbf{Q}_{\mathrm{I}}} & \hat{\mathbf{u}_{r, \mathrm{II}}^{(i)}\mathbf{Q}_{\mathrm{II}}}
	\end{bmatrix}\begin{bmatrix}
			\mathbf{Q}_{\mathrm{I}} & \mathbf{0} \\
			\mathbf{0} & \mathbf{Q}_{\mathrm{II}} \\
			\mathbf{I}_{m_{r_1}} & \mathbf{0} \\
			\mathbf{0} & \mathbf{I}_{m_{r_2}}
	\end{bmatrix}\\
	&=\begin{bmatrix}
		\hat{\mathbf{u}_{r, \mathrm{I}}^{(i)}}\mathbf{Q}_{\mathrm{I}}+\hat{\mathbf{u}_{r, \mathrm{I}}^{(i)}\mathbf{Q}_{\mathrm{I}}} & \hat{\mathbf{u}_{r, \mathrm{II}}^{(i)}}\mathbf{Q}_{\mathrm{II}}+\hat{\mathbf{u}_{r, \mathrm{II}}^{(i)}\mathbf{Q}_{\mathrm{II}}}
	\end{bmatrix}.
\end{aligned}
\label{eq:syndrome}
\end{equation}
We can see from \eqref{eq:syndrome} that $\mathbf{c}_r$ has two parts, and for the first part $\mathbf{c}_{r, \mathrm{I}}$, it is only related to the error happened to \textbf{Type I} patterns, and for the second part $\mathbf{c}_{r, \mathrm{II}}$, it is only related to the error happened to \textbf{Type II} patterns, which completes the proof.
\end{proof}

From the above analyses, the proposed dimension reduced encoder design scheme can not only improve the quality of beam shape by enabling the coupling of two dimensions of RIS, but also enhance the error correction capability of traditional Hamming code, thus further improving the beam training accuracy. 

\subsection{Beam Training Overhead Analysis}

\begin{table}[t!]
	\centering
	\caption{Beam Training Overheads for Different Frameworks}
	\begin{tabular}{|c|c|}
	\hline
		\textbf{Frameworks} & \textbf{Training Overheads} \\ \hline
		Exhaustive beam training & $N_tN_r$ \\ \hline
		Hierarchical beam training & $4\max\left\lbrace\log_2(N_t), \log_2(N_r)\right\rbrace$ \\ \hline
		Coded beam training & $4\max\left\lbrace n_t, n_r\right\rbrace$ \\ 
		\hline
	\end{tabular}
	\label{tb:overhead}
\end{table}

In this subsection, we will analyze the necessary beam training overheads of the traditional exhaustive beam training framework and traditional hierarchical beam training framework and compare them with that of the proposed coded beam training framework. The results are listed in Table~\ref{tb:overhead}.

Specifically, for the traditional exhaustive beam training framework, each possible beam tuple in space should be sequentially explored before determining the best beam tuple. Given the fact that the number of candidate narrow beams is equal to the number of antenna elements, the total beam training overhead of the exhaustive beam training framework should be $N_tN_r$. In this case, when the number of RIS elements is large, an unacceptable beam training overhead will severely limit the system performance. On the other hand, for the traditional hierarchical beam training framework, we need $2\times2=4$ beams at each layer. Since the numbers of layers at BS and RIS are $\log_2(N_t)$ and $\log_2(N_r)$, respectively, the total beam training overhead is $4\max\left\lbrace\log_2(N_t), \log_2(N_r)\right\rbrace$. Through the hierarchical beam training framework, a lot of incorrect angles are excluded at low layers, so the number of necessary beam training overhead is greatly reduced. 

For the proposed coded beam training framework, the codewords are composed of basis patterns and the redundant patterns. Similarly, we need four beams in each layer. Therefore, the necessary beam training overhead for our proposed framework is $4\max\left\lbrace n_t, n_r\right\rbrace$, where $n_t=\log_2(N_t)+m_t$ and $n_r=\log_2(N_r)+m_r$. Since $m$ is the minimum integer that satisfies $2^m-m-1\leq\log_2(N)$, when $N\leq3$, $m_r\leq \log_2(N)$, which means that the proposed scheme will not introduce a large extra beam training overhead compared to hierarchical beam training framework. 

\subsection{Scalability Discussion}

The proposed coded beam training framework is inherently suitable for multi-user scenarios. During the beam training process, the codebook of a certain layer does not depend on the result of the last layer, so no feedback is needed before finishing the beam training. When there are multiple users in the system, the BS transmits different beams sequentially. Each UE measured the received signal of each codeword and conduct the decoding procedure respectively. After transmitting all codewords, each UE obtains its own direction and feeds the result back to the BS sequentially, which can be exploited to generate directional beams for data transmission.

For multi-path scenarios, the proposed framework can only select the path with the strongest power. This is because when two paths are corresponding to two different angular ranges, we still only choose one range at each layer. However, this is not a very big problem in high-frequency systems, since the number of paths is small and the communcation tends to be dominated by the LoS path or the strongest NLoS path. From this perspective, our proposed scheme can still guarantee a relatively good communication performance by estimating the direction of the strongest path.

Furthermore, since the proposed framework can obtain the large-scale channel information such as the AoA and AoD effectively, it can potentially facilitate some channel estimation schemes such as~\cite{zhi2022ris,zhi2022two}, where the uplink pilots are designed according to the large-scale information to improve the SNR during channel estimation.

The proposed framework can also be extended to active RIS assisted communication systems~\cite{zhang2022active,zhi2022active}. For active RIS, since RIS elements can not only manipulate the phase shift, but also manipulate the amplitude, the codeword design will become easier. What's more, thanks to the capability to enhance the reflected signal, the proposed framework can realize accurate beam training over a longer distance.

In practical systems, the locations of BS and RIS tend to be fixed. Based on this assumption, we can also only conduct beam training for RIS-UE link and further reduce the necessary beam training overhead. If the UE has a high mobility, we can also improve the beam training accuracy by exploiting deep learning~\cite{xu2022time}.

\section{Simulation Results}

\begin{table}[t!]
	\centering
	\caption{Simulation Parameters}
	\begin{tabular}{|c|c|}
	\hline
		BS antenna number $N_t$ & 64 \\ \hline
		RIS antenna number $N_{r_1}\times N_{r_2}=N_r$ & $16\times16=256$ \\ \hline
		Central frequency $f_c$ & $28$ GHz \\ \hline
		The distribution of $\phi, \theta$ & $\mathcal{U}(-\pi, \pi)$ \\ \hline
		Threshold $\Delta$ & $0.3$ \\
		\hline
	\end{tabular}
	\label{tb:parameters}
\end{table}

In this section, we evaluate the performance of the proposed coded beam training framework through numerical experiments. The simulation parameters are listed in \textbf{Table~\ref{tb:parameters}} The antenna spacing is set to $d=\frac{c}{2f_c}$. We compare the achievable rate performance of the proposed coded beam training framework with both the traditional exhaustive beam training framework and the traditional hierarchical beam training framework. The achievable rate is obtained by 
\begin{equation}
	R=\log_2\left(1+\frac{P_t\alpha^2}{\sigma^2}\mathbf{h}_r\mathrm{diag}(\mathbf{v})\mathbf{G}\mathbf{w}\mathbf{w}^H\mathbf{G}^H\mathrm{diag}(\mathbf{v}^H)\mathbf{h}_r^H\right),
	\label{eq:rate}
\end{equation}
where $P_t$ denotes the transmission power at BS, $\alpha=\alpha^G\alpha^r$ is the product of BS-RIS path gain and RIS-UE path gain and $\sigma^2$ denotes the noise power. The reflecting vector of RIS $\mathbf{v}$ and the beamforming vector of BS $\mathbf{w}$ are both determined through the corresponding beam training frameworks. During beam training, the SNR is defined as $\frac{P_t\alpha^2}{\sigma^2}$, which implies the transmission power, the path gain, and the noise power. For a certain SNR, we normalize the channel matrix to eliminate the impact of transmission power and distance and add an AWGN to the received signal based on the specific SNR. Here, we do not utilized the ratio of received signal power to received noise power as SNR since the beamforming gain of the codewords also affect the performance of different beam training frameworks. If we set the ratio of received singal power to received noise power as a fixed value, the impact of beamforming gain is eliminated, which means no matter how badly our codeword is generated, the SNR is the same, which cannot reflect the performance of different frameworks. On the contrary, when evaluating the achievable rate, we choose a fixed $\frac{P_t\alpha^2}{\sigma^2}$ for all simulations, which is set as 10. This is because we want to merely evaluate the beam training performance by the beamforming gain in~\eqref{eq:rate}.

\begin{figure}[t!]
	\centering
	\includegraphics[width=1\linewidth]{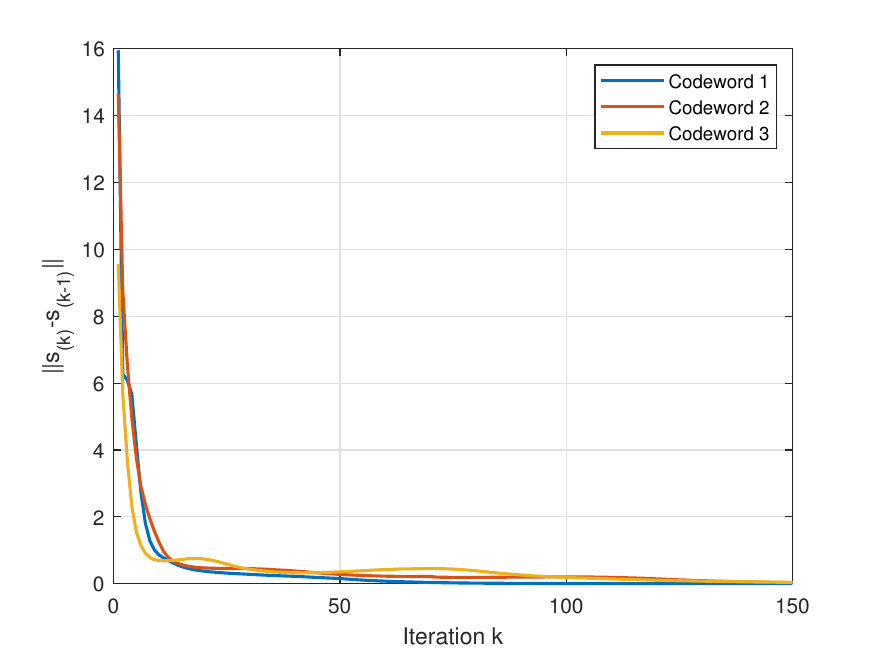}
	\caption{Convergence of the proposed GS-based codeword design scheme.}
	\label{fig:convergence}
\end{figure}

First, we need to determine the number of iterations for the proposed relaxed GS-based codeword design scheme. As illustrated in Fig.~\ref{fig:convergence}, we select three different codewords in the codebook. The convergence is evaluate by $||\mathbf{s}_{(k)}-\mathbf{s}_{(k-1)}||$, which is the difference of the generated beam between the $k$-th round and the $(k-1)$-th round of iterations. We can see that for all codewords, the difference $||\mathbf{s}_{(k)}-\mathbf{s}_{(k-1)}||$ drops fast and approaches 0 at around the $15$-th round. There are also some slight difference in the convergence speed among different codewords, which results from both the random initial phase and the intended beam pattern. In the simulation, we set $K_{\mathrm{iter}}=100$ to guarantee that we have obtained the proper codeword for beam training.

\begin{figure}[t!]
	\centering
	\includegraphics[width=1\linewidth]{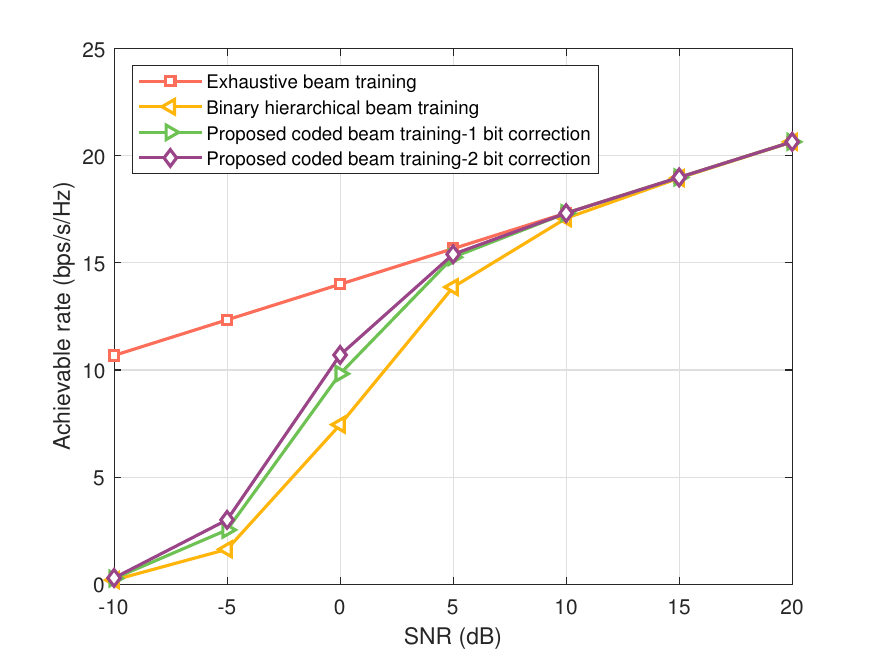}
	\caption{Achievable rate performance v.s. SNR.}
	\label{fig:rate-SNR}
\end{figure}

Fig.~\ref{fig:rate-SNR} depicts the achievable rate performance of different beam training frameworks against the SNR. We assume that the beam training overheads for all frameworks are all sufficient. In this case, the traditional exhaustive beam training framework can always detect the best beam tuples for BS and RIS thanks to the high beamforming gain realized by the large antenna array. We can observe that compared to the traditional hierarchical framework, the proposed coded beam training framework can realize a higher achievable rate performance in low SNR scenarios. The suffix “\textit{1 bit correction}” means that we only utilize the check matrix $\mathbf{H}$ to correct 1-bit error in received codewords as traditional Hamming code, while the suffix “\textit{2 bit correction}” means that we exploit the property of the designed encoder to enable some 2-bit errors to be corrected. We can see that since the proposed dimension reduced encoder can decouple the two dimensions of RIS, the error correction capability can also be enhanced compared to traditional Hamming code. In addition, when the SNR is very low, like -10 dB, the proposed framework coincides with other frameworks. This is because in this case, there are more than 2 bits of errors in the system and our proposed framework cannot correct that many errors. 

\begin{figure}[t!]
	\centering
	\includegraphics[width=1\linewidth]{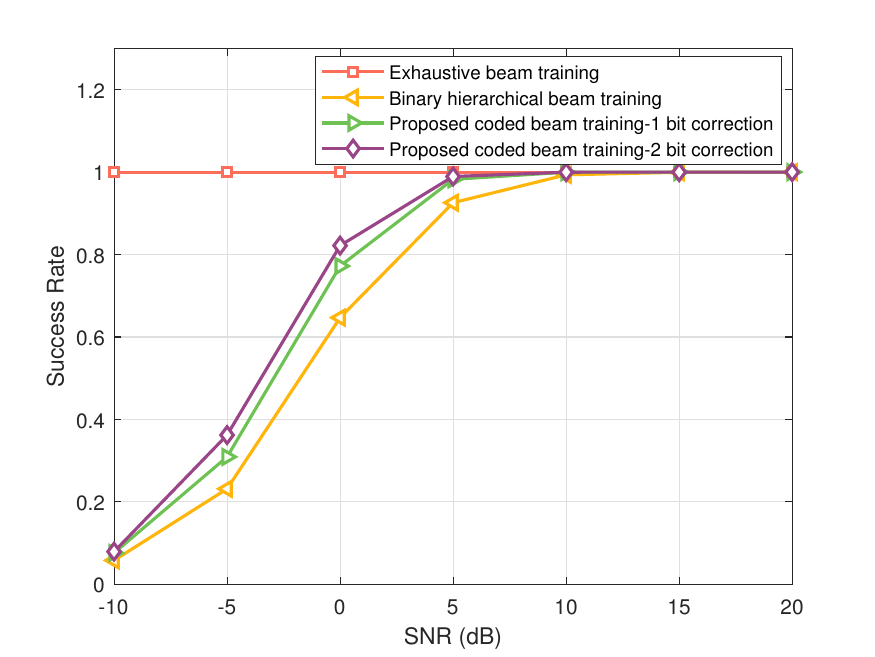}
	\caption{Success rate v.s. SNR.}
	\label{fig:srate-SNR}
\end{figure}

To evaluate the probability of different frameworks to select the best tuple, we compare the success rate of different frameworks against SNR in Fig.~\ref{fig:srate-SNR}. Here, we also assume that the beam training overhead for all frameworks are all sufficient. Similar to the achievable rate performace, the proposed coded beam training framework are more likely to detect the best tuple for BS and RIS successfully compared to traditional hierarchical beam training framework thanks to the error correction capability brought by the encoding-decoding process. In addition, through the decoupling ability enabled by the proposed dimension reduced encoder design scheme, the proposed framework embraces a higher success rate compared to the framework based on traditional Hamming code. 

\begin{figure}[t!]
	\centering
	\includegraphics[width=1\linewidth]{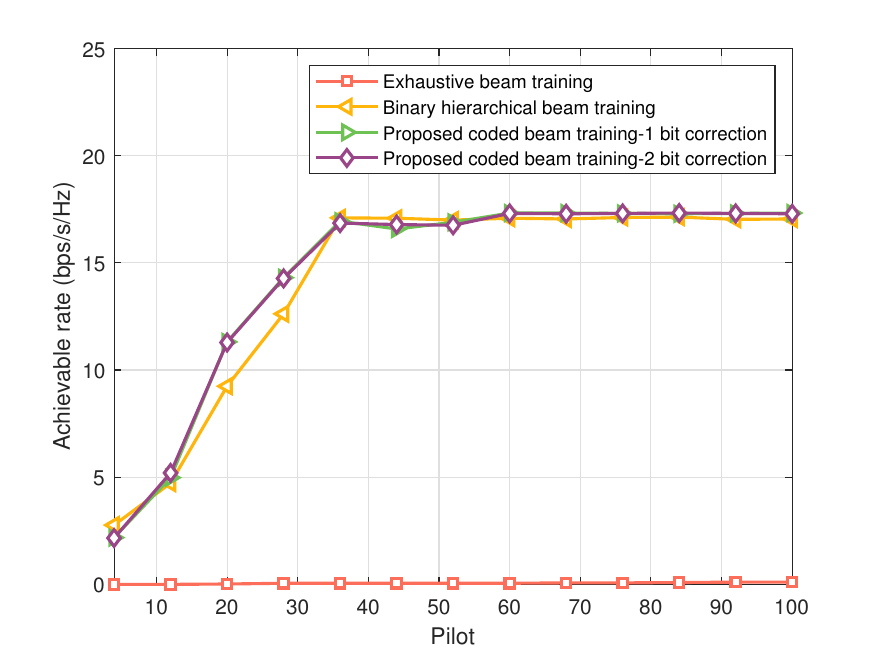}
	\caption{Achievable rate performance v.s. beam training overhead, $P\in\left[4, 100\right]$.}
	\label{fig:rate-pilot-low}
\end{figure}

Furthermore, to reveal the impact of beam training overhead on different frameworks, we compare the achievable rate performance of different beam training frameworks against the number of pilot overheads in Fig.~\ref{fig:rate-pilot-low}. The beam training SNR is set to $10$ dB and the pilot overhead is increasing from $4$ to $100$. In our considered system, the necessary beam training overhead for the traditional hierarchical beam training framework should be $4\max\left\lbrace\log_2(N_t), \log_2(N_r)\right\rbrace=4\times8=32$. The necessary beam training overhead for the proposed coded beam training framework is $4\max\left\lbrace n_t, n_r\right\rbrace=4\times14=56$. We can observe that when the pilot number is insufficient for all frameworks, the proposed coded beam training framework still outperforms existing hierarchical beam training framework since it also have certain error correction capabilities. When the pilot number is sufficient for the hierarchical framework but insufficient for the proposed framework, the achievable rate of the proposed scheme is slightly lower than that of the hierarchical framework. This is because the redundant beams have not been entirely transmitted, so the error correction can sometimes be misleading. When the pilot number is sufficient for the proposed framework, it can reach the maximum achievable rate thanks to the error correction ability. We also notice that in this scenario, the trend of the “1 bit correction” and the “2 bit correction” is nearly the same, this is because when SNR$=10$ dB, both schemes can determine the best beam tuple, which is consistent with the results in Fig.~\ref{fig:rate-SNR} and Fig.~\ref{fig:srate-SNR}. The traditional exhaustive beam training framework, however, can barely detect the best beam tuple since the pilot number is severely insufficient.

\begin{figure}[t!]
	\centering
	\includegraphics[width=1\linewidth]{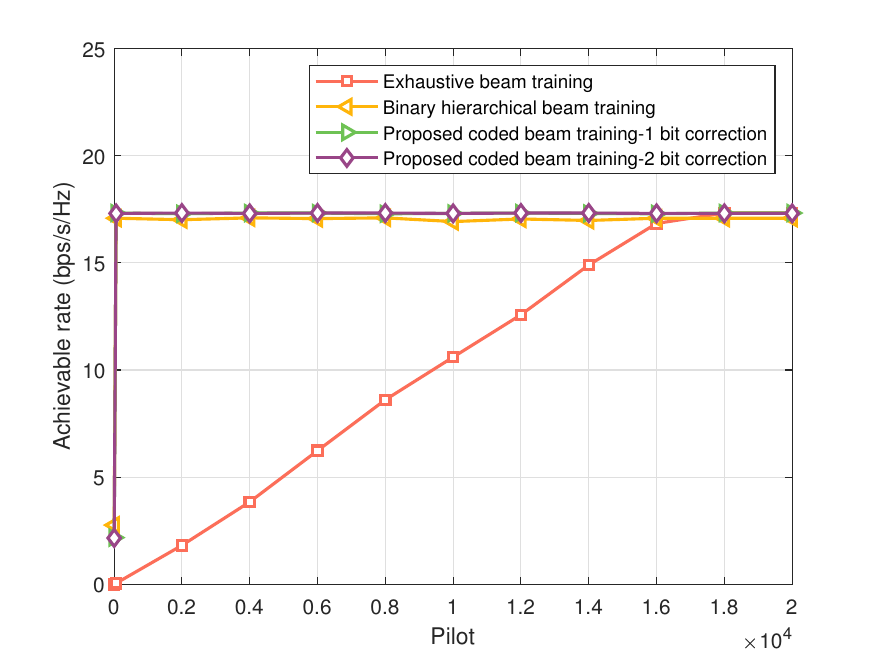}
	\caption{Achievable rate performance v.s. beam training overhead, $P\in\left[4, 20000\right]$.}
	\label{fig:rate-pilot-high}
\end{figure}

To demonstrate the performance improvement of the traditional exhaustive beam training framework more clearly, in Fig.~\ref{fig:rate-pilot-high}, the pilot overhead is increasing from $4$ to $20000$. In our considered scenario, the necessary beam training overhead for the traditional exhaustive beam training framework should be $N_tN_r=16384$. From Fig.~\ref{fig:rate-pilot-high}, we can see that when the pilot number is below $16000$, the achievable rate of the exhaustive framework improves gradually. When the pilot number is around $16000$, the achievable rate is nearly the same as the hierarchical framework. When the pilot number is sufficient, the achievable rate is the same as the proposed coded beam training framework. From Fig.~\ref{fig:rate-pilot-low} and Fig.~\ref{fig:rate-pilot-high}, we can see that due to the error correction capability brought by the encoding-decoding process, the proposed coded beam training framework can outperform existing frameworks under different pilot numbers, which further verified the advantage of the proposed framework. 

\section{Conclusions}

In this paper, we exploited the error correction capability of channel coding to realize accurate beam training under low SNR. By mapping the angles in space to a bitstream, we enabled the encoding-decoding procedure during beam training. Then, considering the constant modulus constraints of RIS elements, we adopted a new codeword design criterion and proposed a relaxed GS-based codeword design scheme. Furthermore, we proposed a dimension reduced encoder design scheme to improve the quality of the beam shape and the capability of error correction simultaneously. Simulation results verified the effectiveness of the proposed scheme. The proposed framework revealed the similarity of intrinsic mathematical structures between channel coding and beam training, which enabled the error correction during beam training and provided a promising solution for accurate and reliable beam training in RIS systems. For future works, this coded beam training framework can be extended to more scenarios such as near-field scenarios. In addition, various channel coding methods can be applied to the proposed framework to enable reliable beam training under low SNR.

\bibliographystyle{IEEEtran}
\bibliography{CBT_RIS,IEEEabrv}
\end{document}